\documentclass[11pt,a4paper]{article}
\usepackage{jheppub}
\usepackage{amsmath,amssymb,amsthm,mathtools}
\usepackage{bm}

\newtheorem{theorem}{Theorem}

\usepackage{graphicx}
\usepackage{url}
\usepackage{float}
\usepackage{tikz}

\newcommand{\ints}{\mathbb{Z}}
\newcommand{\grad}{\nabla}
\newcommand{\bfx}{\mathbf{x}}
\newcommand{\Tr}{\operatorname{Tr}}

\newcommand{\ii}{\mathrm{i}}
\newcommand{\dd}{\mathrm{d}}
\newcommand{\jelp}{j_{|\ell|, p}}
\DeclareMathOperator{\sgn}{sgn}
\DeclareMathOperator{\SF}{SF}
\DeclareMathOperator{\spec}{spec}

\title{$h/2e$ Flux Trapping without Superconductivity}

\author{Will Lancer}
\affiliation{Simons Center for Geometry and Physics, Stony Brook University, Stony Brook, NY}
\emailAdd{william.lancer@stonybrook.edu}

\abstract{%
Recently, it was shown by Komargodski and Popov \cite{Komargodski:2026vwn} 
that ordinary metals threaded with an Aharonov--Bohm solenoid exclusively trap 
either $0$ or $h/2e$ flux. The latter value is familiar in the superconducting
setting, yet we find that it holds in ordinary metals as well, with no condensate
or pairing necessary. We provide an analytic explanation for the flux trapping, and 
extend the results to the cases of finite temperature, time-reversal invariant 
potentials, and lattice regularization, where we find that the flux remains trapped at 
$0$ and $h/2e$.
}

\begin{document}
\maketitle

\section{Introduction}
\label{sec:intro}

In the usual picture of a metal pierced by an Aharonov--Bohm 
solenoid \cite{Aharonov:1959ph}, the magnetic flux is considered to 
be a fixed background for the system. In a recent paper by
Komargodski and Popov \cite{Komargodski:2026vwn}, it was
found that this view of the flux misses an interesting dynamical
effect: if we allow the flux to vary, due to the back-reaction of the 
electrons in the system, the flux is dynamically quantized to either 
be $0$ or $h/2e$. Furthermore, a similar analysis found a near-origin 
persistent current that discontinuously jumped as the flux was turned on.

The flux trapping is an energetic, infrared effect. It is controlled 
by the effective potential obtained from integrating out the fermions, 
which depends on the reduced flux,
\begin{align}
    \label{eq:intro-nu}
    & \nu \coloneqq \frac{\Phi}{\Phi_0},
    & 0 \leq \nu < 1.
\end{align}
The effective potential $\Omega(\nu)$ is periodic mod $1$
in $\nu$ due to large gauge invariance \cite{byers-yang}, and is even in $\nu$
due to time-reversal symmetry. These conditions single out---in
the absence of spontaneous symmetry breaking---the flux
values $\nu = 0$ and $\nu = 1/2$. The value chosen by
the system is then energetically determined.

The value of $\nu = 1/2$ is familiar in the superconducting setting,
but here it arises for a different reason. In a superconductor,
the order parameter $\Psi$ carries charge $2e$, and requiring
it to be single-valued quantizes flux in units of $h/2e$. This means that
the period of the theory is itself $h/2e$. This is not the case here. The
period of the theory remains $h/e$, and $\nu = 0$ and $\nu = 1/2$
are singled out for the reasons given above.

Flux trapping fundamentally occurs because of the back-reaction
of the fermions against the reduced AB flux $\nu$. This back-reaction
produces a current which influences the flux $\Phi$, which then
settles into a trapped value. On a finite disk, we can write
the current as the current on an infinite disk plus a correction
term
\begin{align}
    \label{eq:intro-current}
    j_R(r) = j_\infty(r) + \Delta j(r),
\end{align}
where as it turns out
\begin{align}
    \label{eq:intro-current-correction}
    \Delta j(r) = \frac{c_R(\nu)}{2\pi R} \frac{1}{r}.
\end{align}
Here, $R$ is the radius of the disk and $c_R$ is a constant linear in $\nu$. 
This linear correction to the current is enough to show that flux is trapped 
on the disk, as we demonstrate in Section~\ref{sec:disk-flux-trapping}.

The flux-trapping effect persists for finite temperature---it
is merely damped order by order by the Lifshitz--Kosevich \cite{Lifshitz:1956zz} factor
\begin{align}
    D_\beta(r) = \frac{\pi^2 r / \beta}{\sinh(\pi^2 r / \beta)}.
\end{align}

Interactions do generically change the location of the flux trapping,
but as long as the interactions are time-reversal invariant, the flux
values of $\nu = 0$ and $\nu = 1/2$ are still privileged in some sense.

Flux trapping is an IR effect, and thus it should not depend on the details
of the short-distance physics. To test this, we put the theory on the lattice.
We first write the bulk lattice model, without the solenoid, as
\begin{align}
    \label{eq:intro-lat-ham}
    H_0 = \sum_i h_{ii} c_i^\dagger c_i - \sum_{ij} (t_{ij} c_j^\dagger c_i + t_{ij}^* c_i^\dagger c_j),
\end{align}
where $c_j^\dagger$ and $c_j$ are the creation and annihilation
operators, and $t_{ij}$ is the amplitude for a particle to hop from site $i$
to site $j$. We then add the solenoid as a defect. Away from the solenoid the
vector potential is pure gauge, so it enters only through the Peierls phases,
\begin{align}
    \label{eq:intro-peierls}
    \theta_{ij} = \frac{2\pi}{\Phi_0} \int_{i}^j \dd\mathbf{s} \cdot \mathbf{A}(\mathbf{s}),
\end{align}
which multiply the hopping terms $t_{ij}$. This is minimal 
coupling on the lattice. The full Hamiltonian is
$H = H_0 + \Delta H$, with
\begin{align}
    \Delta H = -\sum_{ij} \left[ (e^{\ii\theta_{ij}} - 1)\, t_{ij} c_j^\dagger c_i
    + (e^{-\ii\theta_{ij}} - 1)\, t_{ij}^* c_i^\dagger c_j \right].
\end{align}
Using the gauge freedom $\theta_{ij} \mapsto \theta_{ij} + \zeta_j - \zeta_i$,
we can move all of the phase onto the bonds that cross a single cut running
from the solenoid to the outer boundary. A bond crossing the cut
counterclockwise carries $\theta_{ij} = 2\pi\nu$, and every other bond
carries $\theta_{ij} = 0$. So, $\Delta H$ is supported on a line, and
the only physical contribution of the defect is the holonomy $\nu$ around the solenoid.
In Section~\ref{sec:lattice} we instead work in the symmetric gauge
$\theta_{ij} = \nu \Delta_{ij}$, which is equivalent.

We can then numerically simulate this model on the lattice
and determine whether or not flux trapping is present. 
We simulate eight different values of the temperature, 
and five different values of the chemical potential, as well as 
five different UV regularizations. This gives $200$ simulations.
In every simulation, either $\nu = 0$ or $\nu = 1/2$ flux was chosen 
by the system.

Flux trapping can be tested in the lab. The cylinder is the feasible target,
as the disk requires phase coherence over the length scale of around $1\,\mathrm{cm}$. The cylinder,
however, is entirely on the micron scale. Specifically, a ballistic copper cylinder
with radius $R = 1\,\mu\mathrm{m}$, height $L = 20\,\mu\mathrm{m}$, and temperature $T = 100\,\mathrm{mK}$
is sufficient to trap flux. To allow for dynamical flux trapping, there are constraints
on the inductance of the cylinder. Roughly, the shorter the cylinder, the softer the
inductance has to be to allow for flux trapping. For the previous example,
an effective inductance of $\sim 0.2\,\mathrm{nH}$ or more---which requires
filling the cylinder with a material of relative permeability $\mu_r \gtrsim 10^3$---is
necessary to trap flux.

One might be concerned that we are ignoring higher-order effects,
like the electron--phonon and phonon--phonon coupling in this analysis,
and that this might affect the final result. As expected, it turns out that in a
weakly coupled metal, the higher-order corrections to the current are small compared
to the current correction $\Delta j(r)$ in (\ref{eq:intro-current}).
Electron--phonon coupling yields fractional corrections $\lambda_{\mathrm{ep}} \sim
0.1$ \cite{savrasov}, where $\lambda_{\mathrm{ep}}$ is the dimensionless
electron--phonon coupling. Phonon--phonon contributions have the expansion parameter
$\epsilon_{\mathrm{pp}} \sim k_B \Theta_D / (B V_a)$, with $B$ the bulk modulus,
$V_a$ the atomic volume, and $\Theta_D$ the Debye temperature \cite{allen}.
For a stiff lattice, $\epsilon_{\mathrm{pp}} \sim 10^{-3}$, and so phonon-mediated
response changes by $\lambda_{\mathrm{ep}} \epsilon_{\mathrm{pp}}
\sim 10^{-4}$--$10^{-3}$. Thus, phonon--phonon effects change the 
linear current correction by at most $10^{-3}$, while electron--phonon 
coupling changes its coefficient by $\lambda_{\mathrm{ep}} \sim 10\%$, 
which we estimate below when we discuss the magnetic enclosure. Neither
effect changes the sign of the restoring current, so flux trapping is preserved.

The photon contribution splits into two pieces. The first is the mean field sourced 
by $\langle \mathbf{J} \rangle$, which is the back-reaction that traps the flux in 
the first place. This is already included in our analysis through Amp\`ere's law, 
(\ref{eq:ampere}). We ignore the second contribution in our analysis, which is the 
electron--electron interaction term $\mathbf{J} \cdot \mathbf{J}$.
Its bare coupling is $\alpha v_F/c = \mu_0'/4\pi \sim 10^{-5}$--$10^{-4}$ \cite{loss-martin},
so it changes $\Omega(\nu)$---and thus the restoring-current coefficient---an
amount on the order of $10^{-5}$--$10^{-4}$. This is negligible compared to the
electron--phonon contribution and does not change the flux-trapping result. It should 
also be noted that $\mathbf{J} \cdot \mathbf{J}$ is time-reversal 
invariant, and thus it cannot move the trapped flux from $\nu = 0$ or $\nu = 1/2$.

Such small corrections do not qualitatively change flux trapping unless the
magnetic enclosure induces a large IR enhancement. The enclosure could strengthen
the magnetic field produced by current fluctuations, making the photon-mediated
correction much larger than the bare estimate---we do not expect this to occur
for the copper setup above. The dominant corrections to the
restoring-current coefficient are of order $\sim 10\%$, and the sign
of the restoring current is preserved. The trapping length
on the disk depends on the restoring-current coefficient by an inverse cube
root, so these corrections only mildly shift the scale. For example, in the
effective-mass estimate at fixed $k_F$, $m^*/m_b = 1 + \lambda_{\mathrm{ep}}$
gives $\ell_{\mathrm{eff}}/\ell_M = (1 + \lambda_{\mathrm{ep}})^{1/3}$.
Letting $\lambda_{\mathrm{ep}} = 0.14$ gives a $4.5\%$ increase.
For the cylinder used in the experimental setup, the dependence is an inverse
square root, so $\ell_{z, \mathrm{eff}}/\ell_z =
(1 + \lambda_{\mathrm{ep}})^{1/2}$ gives a $6.8\%$ increase.
This implies that our leading analysis is qualitatively
correct for describing flux trapping in a weakly coupled metal.

The near-origin current is quite different from the flux trapping.
It is a UV effect, and is localized at radii $r \lesssim k_F^{-1} e^{-1/|\nu|}$,
therefore being outside the range of the effective theory of the Fermi surface.
Thus, the near-origin current is non-universal, and depends closely on the
precise details of the UV physics. In the continuum model, the persistent current
is controlled by the spectral information of the circle Dirac operator,
$D(\nu) = p_\theta - \nu$. In Appendix~\ref{app:uv-current}, we explore
these calculations, consider the case of finite temperature, and add radial potentials
to the Hamiltonian.

\section{Trapping Flux on the Cylinder}
\label{sec:cylinder}

Here we detail the general picture we will be working with.
The goal of this section is to analyze flux trapping on the cylinder.
We will first cover some general theorems about flux trapping that apply to
both the cylinder and the disk, then present a simple selection rule for which
value of the flux the cylinder traps, and then extend the results to finite temperature.

We will only deal with annular geometries in this paper, as we must pierce 
the metal with an Aharonov--Bohm solenoid. Specifically, we will be working with the case of a cylinder, 
as imaged in Figure~\ref{fig:just-cyl}. Komargodski and Popov \cite{Komargodski:2026vwn}
presented numerics that show flux trapping at $\nu = 0$ and $\nu = 1/2$
in the case of the disk, and we do not reproduce these considerations here.
Due to the piercing of the geometries by the AB solenoid,
there is naturally a gauge field present, $A$. We also take the
metals to have possibly interacting spinless electrons living on them, $\psi$.
These electrons back-react in response to the induced AB flux,
producing the dynamic flux quantization we see in the examples below.

If our system has a gauge field, $A$, and the space we're working
on is not simply connected, it is natural to consider the holonomy
of the gauge field around the origin,
\begin{align}
    \label{eq:holonomy}
    \Phi = \oint A = \int_0^{2\pi} A_\theta \, \dd\theta.
\end{align}
We will find it simpler to work with the flux measured in units
of the flux quantum, $\Phi_0 = h/e$,
\begin{align}
    & \nu \coloneqq \frac{\Phi}{\Phi_0},
    & 0 \leq \nu < 1.
\end{align}
Since integer holonomies can be gauged away, $\nu$ is defined
mod $1$. In this paper, we will generally find it beneficial to
measure things in terms of the Fermi momentum---thus, we switch
to dimensionless variables, measured as
\begin{align}
    \label{eq:dimensionless-vars}
    x = k_F^{-1} x',
    \hspace{0.25in}
    \hat{\psi} = k_F \hat{\psi}',
    \hspace{0.25in}
    \mathbf{A} = \frac{\hbar k_F}{e} \mathbf{A}',
    \hspace{0.25in}
    \mathbf{J} = \frac{e \hbar k_F^4}{m} \, \mathbf{J}',
    \hspace{0.25in}
    E = \frac{\hbar^2 k_F^2}{2m} E',
\end{align}
where $\mathbf{J}$ is the three-dimensional electric current of the system, and
$\mathbf{A}$ the gauge field. For example, in these units, $E_F' = 1$,
and Amp\`ere's law is
\begin{align}
    \label{eq:ampere}
    \grad' \times \mathbf{B}' = \mu_0' \mathbf{J}',
\end{align}
where $\mu_0' \coloneqq 4\pi \alpha v_F/c$. In a real
metal, $\mu_0' \ll 1$. This corresponds to our gauge field
being very stiff, as $\mu_0' \ll 1$ makes the energy cost
of changing the gauge field very large. In two dimensions,
the dimensionless current is instead
\begin{align}
    \mathbf{J} = \frac{e \hbar k_F^3}{m} \, \mathbf{J}'.
\end{align}
The situation we're in should be clear from context. From now on, 
we will drop the primes on our variables.

Physically, we are working with two-dimensional objects in three-dimensional
space. We choose to encase our setups in a high-permeability material---this
kills off any tangential $\mathbf{B}$-fields, thus fixing 
$\hat{\mathbf{n}} \times \mathbf{B} = \mathbf{0}$ on the boundary. We will work 
exclusively with this setup in this paper.

In the example of the disk below, we aim to show that
the value of $\nu$ is dynamically quantized to be either $0$ or $1/2$.
This will basically be a consequence of the energy
of the system being periodic mod $1$ along with time-reversal
symmetry, as we will show. We will also show that the flux-trapping effect is 
robust to finite temperature.

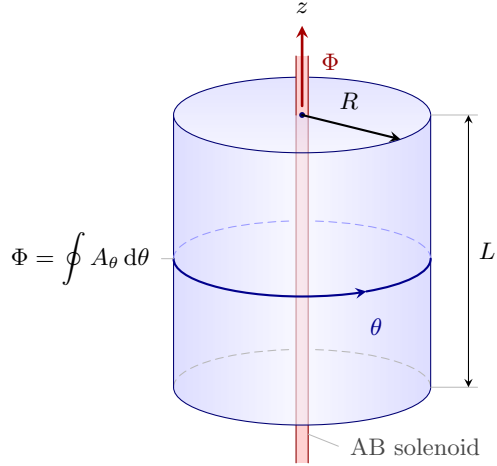
\begin{figure}[H]
    \centering
    \begin{tikzpicture}[>=stealth, line cap=round, line join=round,
        nlab/.style={font=\footnotesize}]
        \def\rx{1.7}   
        \def\ry{0.5}   
        \def\h{3.6}    
        \def\ym{1.7}   
        \def\tubeout{0.78}  
        \def\arrowtop{1.18} 

        \fill[red!18] (-0.08,-1.0) rectangle (0.08,\h);
        \draw[red!55!black, line width=0.4pt] (-0.08,-1.0) -- (-0.08,\h);
        \draw[red!55!black, line width=0.4pt] ( 0.08,-1.0) -- ( 0.08,\h);

        \begin{scope}[opacity=0.60]
            \shade[left color=blue!22, right color=blue!22, middle color=blue!4]
                (-\rx,\h) arc (180:360:{\rx} and {\ry}) -- (\rx,0)
                arc (360:180:{\rx} and {\ry}) -- cycle;
        \end{scope}
        \draw[gray!55, densely dashed] (-\rx,0) arc (180:96:{\rx} and {\ry});
        \draw[gray!55, densely dashed] ({\rx*cos(84)},{\ry*sin(84)})
            arc (84:0:{\rx} and {\ry});
        \draw[blue!45!black] (-\rx,0) arc (180:360:{\rx} and {\ry});
        \draw[blue!45!black] (-\rx,0) -- (-\rx,\h);
        \draw[blue!45!black] ( \rx,0) -- ( \rx,\h);
        \begin{scope}[opacity=0.60]
            \shade[top color=blue!18, bottom color=blue!6] (0,\h) ellipse ({\rx} and {\ry});
        \end{scope}
        \draw[blue!45!black] (0,\h) ellipse ({\rx} and {\ry});

        \fill[red!18] (-0.08,\h) rectangle (0.08,\h+\tubeout);
        \draw[red!55!black, line width=0.4pt] (-0.08,\h) -- (-0.08,\h+\tubeout);
        \draw[red!55!black, line width=0.4pt] ( 0.08,\h) -- ( 0.08,\h+\tubeout);

        \draw[blue!40, densely dashed] (-\rx,\ym) arc (180:96:{\rx} and {\ry});
        \draw[blue!40, densely dashed] ({\rx*cos(84)},{\ym+\ry*sin(84)})
            arc (84:0:{\rx} and {\ry});
        \draw[blue!55!black, thick, ->] (-\rx,\ym) arc (180:300:{\rx} and {\ry});
        \draw[blue!55!black, thick] ({\rx*cos(300)},{\ym+\ry*sin(300)})
            arc (300:360:{\rx} and {\ry});
        \node[nlab, blue!55!black, anchor=north] at (0.98,1.05) {$\theta$};
        \node[nlab, anchor=east] at (-\rx-0.18,\ym)
            {$\displaystyle \Phi = \oint A_\theta \, \dd\theta$};
        \draw[gray!55, line width=0.3pt] (-\rx-0.16,\ym) -- (-\rx,\ym);

        \draw[->, red!65!black, line width=1.0pt] (0,\h+0.12) -- (0,\h+\arrowtop);
        \node[nlab, red!60!black, anchor=west] at (0.12,\h+0.7) {$\Phi$};
        \node[nlab, anchor=south] at (0,\h+\arrowtop+0.03) {$z$};

        \fill[blue!45!black] (0,\h) circle (1.1pt);
        \draw[->, thick] (0,\h) -- ({\rx*cos(320)},{\h+\ry*sin(320)});
        \node[nlab, anchor=south] at (0.62,\h-0.06) {$R$};

        \draw[gray!55, line width=0.3pt] (\rx,0) -- (\rx+0.62,0);
        \draw[gray!55, line width=0.3pt] (\rx,\h) -- (\rx+0.62,\h);
        \draw[<->] (\rx+0.5,0) -- (\rx+0.5,\h);
        \node[nlab, anchor=west] at (\rx+0.5,\h/2) {$L$};

        \node[nlab, gray!55!black, anchor=west] at (0.5,-0.78) {AB solenoid};
        \draw[gray!55, line width=0.3pt] (0.48,-0.78) -- (0.08,-0.62);
    \end{tikzpicture}
    \caption{The cylinder setup. We thread a thin AB solenoid through
    a cylinder of radius $R$ and length $L$, with $L \gg R$. On the cylinder live spinless
    electrons, $\psi$. The solenoid produces a flux $\Phi$ felt by the electrons, given by
    the holonomy of its gauge field $\mathbf{A}$, $\Phi = \oint A = \int_0^{2\pi} A_\theta \, \dd\theta$.}
    \label{fig:just-cyl}
\end{figure}

\subsection{Trapped Flux and Time-Reversal Invariance}
\label{sec:trapped-flux}

We now state and prove some general theorems that tell us that
flux trapping at $\nu = 0$ and $\nu = 1/2$ is a priori expected
for time-reversal invariant systems. In the case of spontaneous symmetry 
breaking, the ground state of the system sits at one of the values in the 
pair $\{ \nu_*, - \nu_* \}$ with $\nu_*$ taken mod $1$.

\begin{theorem}
    Let $V(x)$ be time-reversal invariant, and suppose the ground-state energy
    $E(\nu)$ of $H(\nu) + V(x)$ is differentiable at $\nu \in \{0, 1/2\}$. 
    Then $\nu = 0$ and $\nu = 1/2$ are critical points of $E(\nu)$.
\end{theorem}

\begin{proof}
    First, consider the fact that the energy---and thus the effective potential 
    from integrating out the fermions---is periodic mod $1$ in $\nu$. This follows 
    from the fact that $\nu$ itself is periodic mod $1$, as integer holonomies may be 
    gauged away. So, if $E$ denotes the energy of our system,
    \begin{align}
        \label{eq:byers-yang}
        E(\nu) = E(\nu + 1).
    \end{align}
    In superconductors, this fact is known as Byers--Yang invariance \cite{byers-yang}.
    In ordinary metals, $E(\nu) = E(\nu + 1)$ follows from the
    fact that the spectrum of the Hamiltonian remains fixed after conjugating
    with the large gauge transformation $U = \exp(\ii \sum_j \theta_j n_j)$,
    where $\theta_j$ is the polar angle of site $j$ and $n_j$ its number operator.
    Now, we claim that if the Hamiltonian of our system, $H$, is time-reversal
    invariant, then $\nu = 0$ and $\nu = 1/2$ are critical points of the energy.
    This follows from the fact that time-reversal is an anti-unitary symmetry,
    and so it acts like
    \begin{align}
        & \mathcal{T} p \mathcal{T}^{-1} = - p,
        & \mathcal{T} V(x) \mathcal{T}^{-1} = V(x).
    \end{align}
    In the regime of minimal coupling, the flux enters the Hamiltonian through a term like
    $(p_\theta - \nu)^2$. Acting on such a term with $\mathcal{T}$ gives
    \begin{align}
        \mathcal{T} (p_\theta - \nu)^2 \mathcal{T}^{-1} = (-p_\theta - \nu)^2 = (p_\theta + \nu)^2,
    \end{align}
    and so
    \begin{align}
        \label{eq:H-is-even}
        \mathcal{T} H(\nu) \mathcal{T}^{-1} = H(-\nu).
    \end{align}
    Combining this with (\ref{eq:byers-yang}) gives the conclusion.
\end{proof}

Note that this conclusion only holds because $E(\nu)$ is differentiable
at the points $\nu = 0$ and $\nu = 1/2$---differentiability fails in the case
of the persistent current in Appendix~\ref{app:uv-current}.

We now show that the only way for a system to sit at a point 
besides $\nu \in \{ 0, 1/2 \}$ is for $\mathcal{T}$-symmetry 
to be spontaneously broken.

\begin{theorem}
    \label{thm:0-or-1/2}
    For a time-reversal invariant Hamiltonian depending on an external flux
    $\nu$, one of the two possibilities holds:
    \begin{enumerate}
        \item The minima of the system lie at $\nu = 0$ or $\nu = 1/2$.
        \item $\mathcal{T}$-symmetry is spontaneously broken, and the minimum
        of the system is one of the values in the pair $\{ \nu_*, - \nu_* \}$ 
        with $\nu_*$ taken mod $1$.
    \end{enumerate}
\end{theorem}

\begin{proof}
    Let $M$ be the set of minima of the system. Since the energy is even, $M = - M$.
    Thus, if $\nu_*$ is a minimum of $E(\nu)$, then $- \nu_*$ is as well. So the minima
    of $E$ appear in pairs, $\{ \nu_*, - \nu_* \}$. These pairs fall into two categories:
    either $\nu_*$ is fixed by the map $\nu_* \xrightarrow{\mathcal{T}} - \nu_*$, or
    it's not. The fixed points of this map are $\nu_* \in \{ 0, 1/2 \}$, and thus
    these minima are invariant under time-reversal symmetry. If the point $\nu_*$
    is not fixed, then time-reversal symmetry is spontaneously broken, as the ground
    state of the system is not $\mathcal{T}$-invariant, but the full theory is.
\end{proof}

So we can see that if we add interactions to the system, either the minima
stay at $\{ 0, 1/2 \}$ or else $\mathcal{T}$-symmetry must be broken.
Is there a way to know which minima will be chosen prima facie? The answer
to this is unknown, but a conjecture by Leggett states that this only depends on the
parity of the particle number of the system \cite{leggett}. For
spinless fermions in one dimension, this result has been proven by Loss \cite{loss}.

\subsection{Zero Temperature}

\label{sec:zero-temp-cyl}

We work on the geometry of the cylinder, as pictured in
Figure~\ref{fig:just-cyl}. The system is axially symmetric, and thus
$\partial_\theta \mathbf{A}(\theta, z) = \mathbf{0}$. To account 
for the back-reaction effects of the fermions, we consider the 
energy on the cylinder:
\begin{align}
    \label{eq:cyl-energy}
    E[A_\theta(z), \psi] = \int R \, \dd\theta \dd z \, \left( \frac{1}{\mu_0 R^2} (\partial_z A_\theta(z))^2 + \left\langle p_z^2 + \frac{1}{R^2}(p_\theta - A_\theta)^2 - 1 \right\rangle \right),
\end{align}
where the $-1$ just means that we are measuring the energy of the fermions
relative to the Fermi energy $E_F$. Minimizing (\ref{eq:cyl-energy}) with
respect to $\psi^*$ gives the Schr\"odinger equation for $\psi$,
\begin{align}
    \label{eq:cyl-schro}
    \left( - \partial_z^2 + \frac{(p_\theta - A_\theta)^2}{R^2} \right) \psi = E \psi.
\end{align}
Since there is axial symmetry, $\psi(z, \theta) = e^{\ii\ell \theta} \psi_\ell$.
Additionally, with $A_\theta(z) = \nu$, there is translational symmetry
in the $z$-direction, so for a cylinder of length $L$, we have that
\begin{align}
    \label{eq:cyl-psi}
    \psi(z, \theta) = \frac{1}{\sqrt{2\pi R L}} e^{\ii\ell \theta} e^{\ii kz}.
\end{align}
Inputting this into (\ref{eq:cyl-schro}) gives
\begin{align}
    \left( k^2 + \frac{(\ell - \nu)^2}{R^2} \right) \psi = E_{\ell, k} \psi,
\end{align}
and so the energy of the fermions is
\begin{align}
    E_{\ell, k} = k^2 + \frac{(\ell - \nu)^2}{R^2}.
\end{align}
We'd like to integrate out the fermions to determine the location of the 
trapped flux, $\nu$. Doing this just amounts to minimizing
the grand canonical potential at fixed values of the gauge field,
\begin{align}
    E_{\rm eff}[A_\theta(z) = \nu] = \min_{\psi}(E - \mu N).
\end{align}
We will work at zero temperature in this section, so $\mu = E_F$.
In Section~\ref{sec:finite-temp-cyl}, we will let $\mu$ vary. We will find that this
does not change the location of the flux trapping from the set $\nu \in \{ 0, 1/2 \}$. 
At zero temperature, integrating out $\psi$ amounts to filling the Fermi sea.
This means that we must sum over all $\ell$ and $k$ such that
$E_{\ell, k} \leq 1$, so
\begin{align*}
    k^2 + \frac{(\ell - \nu)^2}{R^2} \leq 1.
\end{align*} 
Define the cutoff momentum,
\begin{align*}
    k_* \coloneqq \sqrt{1 - \frac{(\ell - \nu)^2}{R^2}}.
\end{align*}
This is real if and only if $|\ell - \nu| \leq R$. Filling
the Fermi sea gives an energy density, $\mathcal{E}_{\psi}$,
\begin{align}
    \label{eq:E-psi}
    \mathcal{E}_{\psi}(\nu) = \frac{1}{L} E_{\rm eff}[A_\theta(z) = \nu] = \sum_{|\ell - \nu| \leq R} \int_{-k_*}^{k_*} \frac{\dd k}{2\pi} (k^2 - k_*^2) = - \frac{2}{3\pi} \sum_{|\ell - \nu| \leq R} k_*^3.
\end{align}
So, we can define the effective potential
\begin{align}
    \label{eq:omega-eff}
    \Omega(\nu) \coloneqq - \mathcal{E}_{\psi}(\nu) = \frac{2}{3\pi} \sum_{|\ell - \nu| \leq R} \left(1 - \frac{(\ell - \nu)^2}{R^2}\right)^{3/2}.
\end{align}
There are two important features of $\Omega(\nu)$:
\begin{enumerate}
    \item It is symmetric under shifts $\nu \to \nu + 1$,
    as we can re-index $\ell \to \ell - 1$. Practically,
    this occurs because only the combination $|\ell - \nu|$
    appears in the summation, and thus is translationally invariant
    over $\ints$. This is also obvious from the fact that integer
    holonomies can be gauged away.
    \item It is symmetric under $\nu, \ell \to -\nu, -\ell$. This immediately
    singles out two important values of $\nu$---$0$ and $1/2$---as these are
    fixed points of this operation. The function $\Omega(\nu)$ is even around
    these points, $\Omega(\nu) = \Omega(-\nu)$ and $\Omega(1/2 + \nu) = \Omega(1/2 - \nu)$,
    and thus these are also critical points of $\Omega$. Note that this
    conclusion only follows because $\Omega$ is differentiable at these 
    points---we will see that this is not the case for the infinite
    disk, and thus we cannot draw the same conclusion there.
\end{enumerate}

We can understand (\ref{eq:omega-eff}) a bit better by summing its
Fourier dual instead. Consider the function,
\begin{align*}
    f(x) = \begin{cases}
        \frac{2}{3\pi} \left( 1 - \frac{x^2}{R^2} \right)^{3/2} & |x| \leq R\\
        0 & {\rm else}
    \end{cases}.
\end{align*}
Then, our sum (\ref{eq:omega-eff}) turns into
\begin{align}
    \Omega(\nu) = \frac{2}{3\pi} \sum_{|\ell - \nu| \leq R} \left(1 - \frac{(\ell - \nu)^2}{R^2}\right)^{3/2} = \sum_{\ell \in \ints} f(\ell - \nu).
\end{align}
We can then rewrite this using the Poisson summation formula,
which equates the sum of a function over the integers with
the sum of its Fourier transform,
\begin{align}
    \label{eq:poisson-summation}
    & \sum_{\ell \in \ints} f(\ell - \nu) = \sum_{m \in \ints} \widetilde{f}(m) e^{-2\pi \ii m \nu},
    & \widetilde{f}(m) = \int_{-\infty}^{\infty} \dd x \, f(x) e^{-2\pi \ii m x}.
\end{align}
Rewriting (\ref{eq:omega-eff}) using (\ref{eq:poisson-summation}) gives,
to first order,
\begin{align}
    \label{eq:omega-poisson-sum}
    \Omega(\nu) \approx \frac{R}{4} + \frac{J_2(2\pi R)}{\pi^2 R} \cos(2\pi \nu),
\end{align}
where we used the fact that $\widetilde{f}(m)$ is real and even,
and thus (\ref{eq:poisson-summation}) is a cosine series. Each
term in the $m$-series above goes like $m^{-5/2}$, and thus
the first-order approximation is a good one\footnote{
    One might worry that near the zeros of $J_2(2\pi R)$, the $m = 2$ term dominates
    the $\nu$-dependent part of the sum, and thus introduces a new maximal $\nu$-value
    besides $0$ or $1/2$. This is false. Near $R$ such that $J_2(2\pi R) \approx 0$,
    $J_2(4\pi R) > 0$, and so $\cos(4\pi \nu)$ selects for $\nu = 0$ or $\nu = 1/2$
    yet again, \emph{reinforcing} the trapped flux values!
}.

There are two important features in (\ref{eq:omega-poisson-sum}):
\begin{enumerate}
    \item We may immediately conclude that $\nu = 0$ and $\nu = 1/2$ extremize
    $\Omega(\nu)$, depending on the sign of $J_2(2\pi R)$. If $J_2(2\pi R) > 0$,
    then $\nu = 0$ is chosen---if $J_2(2\pi R) < 0$, then $\nu = 1/2$ is chosen.
    \item The second term in (\ref{eq:omega-poisson-sum}) goes like $R^{-5/2}$
    relative to the first term, and thus in the $R \to \infty$ limit, we should
    expect no preferred value of flux trapping. One way to think about this is that
    as the radius of the cylinder becomes very large, the spacing in the angular momentum
    levels becomes very small, $\Delta \ell \sim 1/R$. Thus, the sum over $\ell$ approximately
    becomes an integral over $\ell$, and so shifting by $\nu$ becomes a trivial change of variables.
\end{enumerate}
We can check the prediction from (\ref{eq:omega-poisson-sum}) numerically
to confirm our results. Comparing with the numerical simulations in \cite{Komargodski:2026vwn}
gives perfect agreement, where $\nu_\infty$ is the value of $\nu$ chosen
by the system:
\begin{table}[H]
    \centering
    \begin{tabular}{|c|c|c|}
        \hline
         & Figure $4$ $\nu_\infty$ value & Predicted $\nu_\infty$ value\\
        \hline
        $R = 1.6$ & $\nu = 0$ & $\nu = 0$ \\
        $R = 1.7$ & $\nu = 0$ & $\nu = 0$ \\
        $R = 1.8$ & $\nu = 0$ & $\nu = 0$ \\
        $R = 1.9$ & $\nu = 1/2$ & $\nu = 1/2$ \\
        $R = 2.0$ & $\nu = 1/2$ & $\nu = 1/2$ \\
        $R = 0.125$ & $\nu = 0$ & $\nu = 0$ \\
        \hline
    \end{tabular}
    \caption{The flux $\nu_\infty$ that maximizes $\Omega(\nu)$, read off from Figure~4
    of \cite{Komargodski:2026vwn} and predicted from the sign of $J_2(2\pi R)$
    in (\ref{eq:omega-poisson-sum}).}
\end{table}

\subsection{Finite Temperature}
\label{sec:finite-temp-cyl}

To extend the considerations of the previous sections to finite temperature,
we first must modify our computation of the grand canonical potential in (\ref{eq:E-psi}).
At finite $T$, we instead use
\begin{align}
    \mathcal{Z} = \Tr\left( e^{-\beta(E - \mu N)}\right).
\end{align}
In the occupation number basis, this becomes
\begin{align}
    \mathcal{Z} = \prod_{\ell \in \ints} \prod_{k} \left( 1 + e^{-\beta(E_{\ell, k} - \mu)}\right).
\end{align}
Using
\begin{align}
    \Omega_{\rm gc} = -\beta^{-1} \log(\mathcal{Z})
\end{align}
gives
\begin{align}
    \label{eq:omega-finite-T}
    \Omega_\beta(\nu) \coloneqq - \Omega_{\rm gc} = \sum_{\ell \in \ints} \int_{-\infty}^\infty \frac{\dd k}{2\pi \beta} \log\left( 1 + e^{-\beta(E_{\ell, k} - \mu)}\right).
\end{align}
Define $c_\ell(\nu) = (\ell - \nu)^2/R^2 - \mu$.
Performing the integral (\ref{eq:omega-finite-T}) gives
a series in $\ell$,
\begin{align}
    \Omega_\beta(\nu) = - \frac{1}{\sqrt{4 \pi \beta^3}} \sum_{\ell \in \ints} \operatorname{Li}_{3/2} \left( - e^{-\beta c_\ell(\nu)}\right).
\end{align}
In the $T \to 0$ limit, we reproduce the result (\ref{eq:omega-eff}),
\begin{align}
    \lim_{\beta \to \infty} \Omega_\beta(\nu) \simeq \frac{2}{3\pi} \sum_{|\ell - \nu| \leq R} \left( \mu - \frac{(\ell - \nu)^2}{R^2} \right)^{3/2},
\end{align}
with $\mu = E_F = 1$. We may again perform Poisson resummation on
the sum (\ref{eq:omega-finite-T}). The function $f$ we sum over is
\begin{align}
    f(\ell - \nu) = \int_{-\infty}^{\infty} \frac{\dd k}{2\pi \beta} \log\left( 1 + e^{-\beta(E_{\ell, k} - \mu)}\right).
\end{align}
This function is real and even, so computing the first two terms in the
$m$-series gives
\begin{align}
    & \widetilde{f}(0) = - \frac{R}{2\beta^2} \operatorname{Li}_{2}(-e^{\beta \mu})
    & \widetilde{f}(1) = \frac{R}{2\beta} \int_0^\infty \dd x \, \log\left( 1 + e^{-\beta(x - \mu)} \right) J_0(2\pi R \sqrt{x}).
\end{align}
So, our effective potential becomes
\begin{align}
    \label{eq:omega-finite-T-poisson}
    \Omega_\beta(\nu) \approx \widetilde{f}(0) + 2 \widetilde{f}(1) \cos(2\pi \nu).
\end{align}
For low temperature, we should recover (\ref{eq:omega-poisson-sum}).
This is indeed the case, as for $\mu = 1$ and $\beta \gg 1$, we have
\begin{align}
    \widetilde{f}(0) & \simeq \frac{R}{4} + \frac{\pi^2 R}{12\beta^2} \approx \frac{R}{4},\\
    2\widetilde{f}(1) & \simeq \frac{J_2(2\pi R)}{\pi^2 R} \frac{\pi^2 R / \beta}{\sinh(\pi^2 R / \beta)} \approx \frac{J_2(2\pi R)}{\pi^2 R},
\end{align}
which reproduces (\ref{eq:omega-poisson-sum}). Notice the appearance of the damping factor
\begin{align}
    \label{eq:lk-damping}
    D_\beta(R) = \frac{\pi^2 R / \beta}{\sinh(\pi^2 R/ \beta)}.
\end{align}
This is the standard Lifshitz--Kosevich damping factor for oscillatory finite-temperature
observables \cite{Lifshitz:1956zz}. This essentially appears because most of the
variation of fermionic systems occurs at the Fermi surface, with variation scale
$\beta^{-1}$. So, if we take some zero-temperature observable $\mathcal{O}_0$ and 
produce a finite-temperature observable from it, $\mathcal{O}_T$, this
amounts to averaging the zero-temperature observable against the variation of the occupation
near the Fermi surface.

We can again draw two important conclusions from (\ref{eq:omega-finite-T-poisson}):
\begin{enumerate}
    \item The extrema of $\Omega_\beta(\nu)$ lie at $\nu = 0$ and $\nu = 1/2$,
    and the value that is chosen by the system is dictated by the sign of $\widetilde{f}(1)$.
    If $\widetilde{f}(1) > 0$, then $\nu = 0$ is chosen, and if $\widetilde{f}(1) < 0$,
    then $\nu = 1/2$ is chosen. A plot of how this sign varies as we tune $R$ and 
    $\beta$ is given in Figure~\ref{fig:finite-temp-f1}, from which we can draw
    two conclusions. First, the value of flux trapping is essentially only determined
    by $R$, with periodicity approximately of the Bessel function zeros.
    Second, at high temperature and large $R$, $\widetilde{f}(1)$ is exponentially damped,
    and so approximately neither value of the flux is preferred by the system.
    \item Branching off of the last point, at large $R$, due to the relative scaling of 
    $\widetilde{f}(1)/\widetilde{f}(0) \sim R^{-5/2}$, there is no preferred value of 
    flux trapping as the system gets very large.
\end{enumerate}

\begin{figure}
    \centering
    \includegraphics[width=0.8\textwidth]{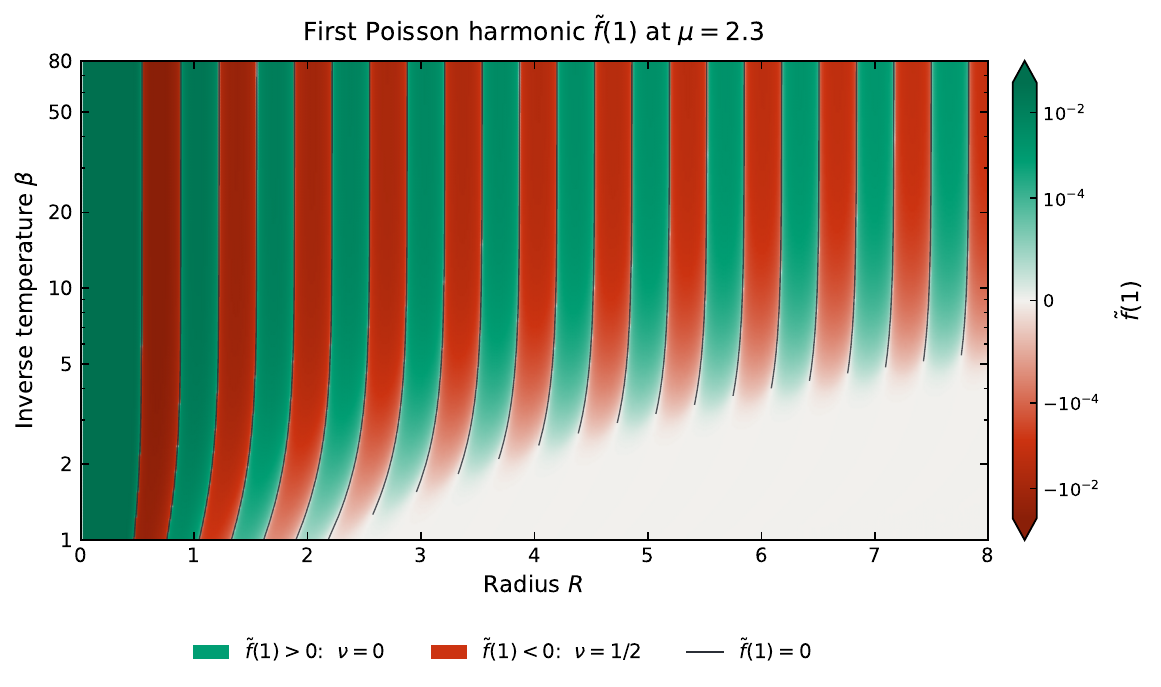}
    \caption{Heatmap of the value of $\widetilde{f}(1)$, plotted against $R$ and $\beta$.
    Red coloration of the graph indicates $\widetilde{f}(1) < 0$, and thus $\nu = 1/2$,
    and green coloration indicates $\widetilde{f}(1) > 0 \iff \nu = 0$. Gray coloration
    indicates $|\widetilde{f}(1)| \lesssim 10^{-6}$, where essentially neither value is preferred
    by the system.}
    \label{fig:finite-temp-f1}
\end{figure}

Restoring units, the Lifshitz--Kosevich factors (\ref{eq:lk-damping})
depend on temperature via
\begin{align}
    \frac{\pi^2 R}{\beta} = \frac{2\pi^2 k_B T R}{\hbar v_F},
\end{align}
so the oscillation persists up to $T_*(R) = \hbar v_F / 2\pi^2 k_B R$.
Note that this condition only depends on the radius of the cylinder 
and not its height. For a copper cylinder with $R = 100$ nm, we get 
$T_* = 6.1$ K, and for $R = 1\,\mu\mathrm{m}$, we get $T_* = 0.61$ K \cite{ashcroft-mermin}.

There are two constraints that must be satisfied for flux trapping to
occur:
\begin{enumerate}
    \item The cylinder must be large enough to trap flux. Expanding
    (\ref{eq:cyl-energy}) around a minimum $\nu_*$ and using
    \begin{align}
        |\Omega''(\nu_*)| = \frac{4 |J_2(2\pi R)| D_\beta(R)}{R},
    \end{align}
    we have
    \begin{align}
        \label{eq:ell-z}
        \ell_z = \left( \frac{4\pi}{\mu_0 R |\Omega''(\nu_*)|} \right)^{1/2}
        \sim \frac{\pi R^{1/4}}{\sqrt{\mu_0}}.
    \end{align}
    This is around $0.11\,\mu\mathrm{m}$ for a copper cylinder with $R = 1\,\mu\mathrm{m}$,
    so for example $L = 20\,\mu\mathrm{m}$ is sufficient to trap flux.
    \item The inductance of the cylinder must allow for flux to dynamically
    trap. Precisely, changing the net flux by $\delta \Phi$ requires driving a
    screening current around the cylinder, which requires an energy of
    $\delta \Phi^2 / 2\mathcal{L}$. The depth $\Delta E$ of the energy well in
    $\Omega$ must be large enough to supply the energy cost of the screening current---in
    the worst case, when the external flux sits at the maximum $\nu_* + 1/2$,
    the screening current requires an energy of
    \begin{align}
        E = \frac{1}{2\mathcal{L}} \left( \frac{h}{2e} \right)^2,
    \end{align}
    and so we must have
    \begin{align}
        \mathcal{L} > \frac{1}{2 \Delta E} \left( \frac{h}{2e} \right)^2.
    \end{align}
    For the copper cylinder above, this amounts to having $\mathcal{L}$ be at least
    $\sim 0.2\,\mathrm{nH}$. In SI units, $\mathcal{L} = \mu_r \mu_0 \pi R^2 / L$,
    where $\mu_r$ is the relative permeability of the material filling the cylinder, as
    in \cite{Komargodski:2026vwn}. An empty cylinder has
    $\mathcal{L} \approx 2 \times 10^{-4}\,\mathrm{nH}$, so we need $\mu_r \gtrsim 10^3$.
\end{enumerate}

In the above we treat the cylinder as a single sheet of spinless electrons. This is
of course a simplification of a realistic experimental setup. Adding spin to our
system doubles $\Omega(\nu)$, and having a wall of finite thickness adds transverse
subbands to the spectrum. These subbands change the depth of the well but not the possible locations
of the flux trapping. We also assume that the electrons are phase coherent over the
cylinder's circumference, $2\pi R \approx 6\,\mu\mathrm{m}$.

\section{Trapping Flux on the Disk}
\label{sec:disk}

In this section, we will demonstrate flux trapping on the disk by studying
linear corrections to the infinite disk current, and then we will show that
flux is trapped when we place the system on the lattice.

On the finite disk, we may take the outer radius of the disk to infinity, 
$b \to \infty$, and then ask what the value of the current density is at radius $r$.
For a range of $r$ values, the infinite disk should be a good approximation to
the finite disk of finite outer radius $b = R$. 
However, we expect at some finite radius $r_*$ for the error term
\begin{align}
    \Delta j^\theta(r) = j_R^\theta(r) - j_\infty^\theta(r)
\end{align}
to become $O(1)$ relative to the infinite current, 
where $j_R^\theta(r)$ is the current density at finite
outer radius $b = R$ and $j_\infty^\theta(r)$ is the current density with $b \to \infty$. 
In \cite{Komargodski:2026vwn}, there is a numerical argument that this occurs at
\begin{align}
    r \sim R^{1/3}.
\end{align}
We now explicitly show this to be the case,
and determine that
\begin{align}
    r_* = \left( \frac{|2\nu - 1| \sin(\pi \nu)}{4\pi |c_R(\nu)|} \right)^{1/3} R^{1/3},
\end{align}
where $c_R(\nu)$ is a constant linear in $\nu$ and depending 
on the fractional part of $R/\pi$.

Past the value of $r_*$, the current correction dominates,
so $j^\theta(r) = j_\infty(r) + \Delta j(r) \simeq \Delta j(r)$.
As we will show in Section~\ref{sec:disk-breakdown},
the current correction goes like $1/r$, and is precisely
\begin{align}
    \Delta j(r) = \frac{c_R(\nu)}{2\pi R} \frac{1}{r}.
\end{align}
This correction to the current is enough to show that flux 
is trapped on the disk, as we demonstrate in Section~\ref{sec:disk-flux-trapping}.

\subsection{Infinite Disk Approximation Breakdown}
\label{sec:disk-breakdown}

First, exchange the azimuthal current for the scaled variable
$Q(r) = r^2 j^\theta(r)$. Implement Dirichlet boundary conditions
at $b = R$, and implement the flux as a twist,
\begin{align}
    \psi(\theta + 2\pi) = e^{-2\pi \ii \nu} \psi(\theta).
\end{align}
Then the wave functions for $b = R$ are given by
\begin{align}
    \psi_{\ell, p}(r) = e^{\ii \ell \theta} R(r) = e^{\ii \ell \theta} \frac{J_{|\ell|}(j_{|\ell|, p} (r/R))}{\sqrt{\pi} R |J_{|\ell| + 1}(j_{|\ell|, p})|},
\end{align}
where $\ell \in \ints - \nu$. Computing $j_R^\theta = (1/r^2) \operatorname{Im}(\psi^* \partial_\theta \psi)$
from this and filling the Fermi sea gives
\begin{align}
    \label{eq:Q-r}
    Q_R(r, \nu) = \frac{1}{\pi R^2} \sum_{\ell \in \ints - \nu} \ell \sum_{j_{|\ell|, p} < R} \left( \frac{J_{|\ell|}(j_{|\ell|, p} (r/R))}{|J_{|\ell| + 1}(j_{|\ell|, p})|} \right)^2.
\end{align}
As computed in Appendix~\ref{app:infinite-R},
\begin{align}
    \label{eq:Q-infty}
    Q_\infty(r, \nu) = \frac{1}{2\pi} \sum_{\ell \in \ints - \nu} \ell \int_0^1 \dd k \, k J_{|\ell|}(kr)^2.
\end{align}
So we have the explicit error term of
\begin{align}
    \label{eq:Q-error}
    \Delta Q(r, \nu) = Q_R(r, \nu) - Q_\infty(r, \nu).
\end{align}
The error here essentially arises from approximating the sum (\ref{eq:Q-r})
by the integral (\ref{eq:Q-infty}). We can take the difference
(\ref{eq:Q-error}) by using the Euler--Maclaurin formula and by making
a number of approximations. Each approximation reduces the range of validity
of the formula (\ref{eq:Q-error}). The details of the calculation (\ref{eq:Q-error})
are in Appendix~\ref{app:finite-R}. The result of this calculation
is that
\begin{align}
    \label{eq:full-delta-Q}
    \Delta Q(r, \nu) = \frac{c_R}{2\pi} \frac{r}{R} + O\left( \frac{r^{2/3}}{R} \right) + o(r/R),
\end{align}
where the first error term was derived analytically, and the second
only supported numerically. This approximation is valid for $R^{1/3} \lesssim r \ll R^{1/2}$.
We plot this for multiple values of $R$ in Figure 
\ref{fig:current-correction-plots}.

In Appendix~\ref{app:infinite-R}, we exactly compute (\ref{eq:Q-infty}). It is given by
\begin{align}
    \label{eq:q-infty-exact}
    Q_\infty(r, \nu) = \frac{(2\nu - 1) \sin(\pi \nu) \sin(2r)}{8\pi^2 r^2} + O(r^{-3}).
\end{align}
Thus, equating (\ref{eq:full-delta-Q}) to the envelope of
(\ref{eq:q-infty-exact}) gives the crossover radius,
\begin{align}
    r_* = \left( \frac{|2\nu - 1| \sin(\pi \nu)}{4\pi |c_R(\nu)|} \right)^{1/3} R^{1/3}.
\end{align}

\begin{figure}
    \centering
    \includegraphics[width=0.98\textwidth]{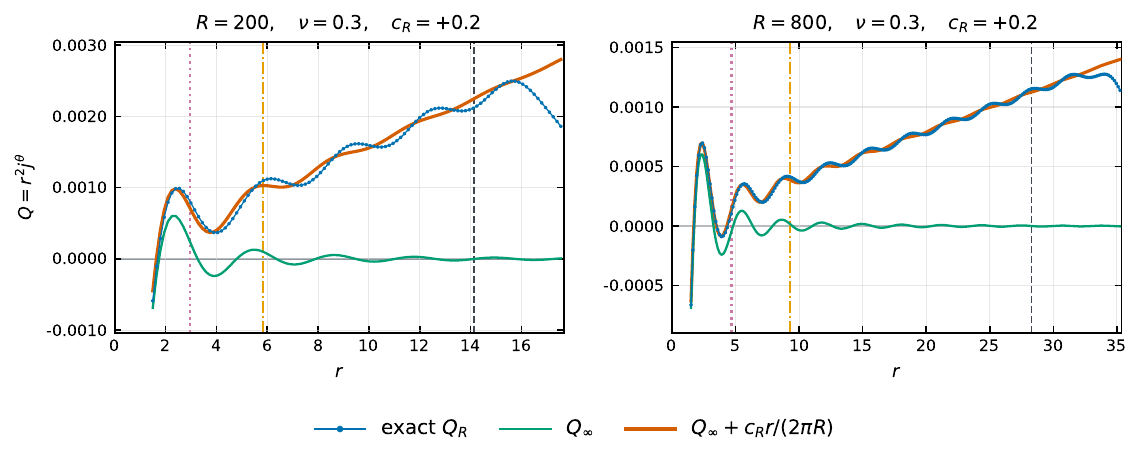}
    \caption{Graphs of the three different currents for two values of the outer radius.
    The purple dotted line is $r_*$, the orange dash-dotted line is $R^{1/3}$,
    and the black dashed line is $R^{1/2}$. The dotted blue line is the exact
    finite-$R$ current, $Q_R(r)$, the green line is the infinite-$R$
    current $Q_\infty(r)$, and the orange line is the infinite-$R$ current plus the
    correction term $\Delta Q(r) = c_R r / 2\pi R$. }
    \label{fig:current-correction-plots}
\end{figure}

\subsection{Trapping Flux on the Disk}
\label{sec:disk-flux-trapping}

Consider the energy on the disk,
\begin{align}
    \label{eq:disk-energy-functional}
    E[\nu(r)] = \int \dd r \left( \frac{2\pi}{\mu_0 r} (\partial_r \nu(r))^2 - \Omega(\nu(r), r) \right),
\end{align}
where $\Omega(\nu, r)$ is the effective potential obtained from
integrating out the fermions. As before,
\begin{align}
    \Omega(\nu, r) = - \Omega_{\rm gc}(\nu, r),
\end{align}
where $\Omega_{\rm gc}$ is the grand canonical potential. $\Omega$
obeys
\begin{align}
    \frac{\delta \Omega}{\delta \nu(r)} = 4\pi r \, j^\theta(r).
\end{align}
Varying (\ref{eq:disk-energy-functional}) gives
\begin{align}
    \label{eq:nu-eom}
    \nu'' - \frac{\nu'}{r} = - \mu_0 r^2 j^\theta(r).
\end{align}
Consider $R^{1/3} \ll r \ll R^{1/2}$. In this regime,
the current is dominated by the linear correction,
\begin{align}
    j^\theta(r) = j_\infty^\theta(r) + \Delta j^\theta(r) \simeq \Delta j^\theta(r),
\end{align}
where
\begin{align}
    \Delta j^\theta(r) = \frac{c_R(\nu(r))}{2\pi R} \frac{1}{r},
\end{align}
as before. Write the holonomy as $\nu(r) = \nu_* + \delta \nu(r)$, where $\nu_*$
is a fixed point, $\nu_* \in \{0, 1/2\}$. We find that
$c_R(\nu_* + \delta \nu(r)) = - \delta \nu(r)$. Substituting this
into equation (\ref{eq:nu-eom}) gives
\begin{align}
    \delta \nu'' - \frac{\delta \nu'}{r} - \frac{\mu_0}{2\pi R} r \, \delta \nu = 0.
\end{align}
This is solved by modified Bessel functions. Discarding a
divergent solution gives the answer
\begin{align}
    \label{eq:delta-nu-decay}
    \delta \nu(r) = C_K r K_{2/3}\!\left( \frac{2}{3} \left( \frac{r}{\ell_M} \right)^{3/2} \right),
\end{align}
where $C_K$ is a coefficient and $K_{2/3}$ a modified
Bessel function of the second kind. We also introduced
the length scale
\begin{align}
    \ell_M = \left( \frac{2 \pi R}{\mu_0} \right)^{1/3}.
\end{align}
Expanding this for large $r/\ell_M$ gives
\begin{align}
    \delta \nu(r) \simeq C_K \frac{\sqrt{3\pi}}{2} \, \ell_M \left( \frac{r}{\ell_M} \right)^{1/4}
    \exp\left( - \frac{2}{3} \left( \frac{r}{\ell_M} \right)^{3/2} \right).
\end{align}
The trapping length constrains the stiffness of the gauge field.
We work in the regime of $r_* \ll r \ll R^{1/2}$, and the decay
(\ref{eq:delta-nu-decay}) is visible only if $\ell_M$ obeys
\begin{align}
    r_* \ll \ell_M \ll R^{1/2} \iff \frac{2\pi}{R^{1/2}} \ll \mu_0 \ll \frac{2\pi}{C_*^3},
\end{align}
where
\begin{align}
    r_* = C_* R^{1/3} = \left( \frac{|2\nu - 1| \sin(\pi \nu)}{4\pi |c_R(\nu)|} \right)^{1/3} R^{1/3}.
\end{align}
In a real metal, $\mu_0 \ll 1$, so the upper bound automatically holds.
Additionally, for large $R$ the lower bound is satisfied. The lower bound
dictates the minimum size of the disk, $R \gg (2\pi/\mu_0)^2$. This
bound is generically quite large, as $\mu_0 = 4\pi \alpha v_F/c$ is, in real
metals, around $5 \times 10^{-4}$ \cite{ashcroft-mermin}, giving
$R \gg 10^8$ in units of $k_F^{-1}$. In physical units, this is a disk
of a radius around a centimeter. This large size makes sense, as flux trapping is an infrared
effect.

\subsection{Lattice Regularization}
\label{sec:lattice}

We now see what happens when we replace the continuum of
the infinite disk with a lattice. Due to the fact that flux trapping
is an IR effect, we do not expect
it to be modified by a lattice regularization. Let us recall some basic
details about the lattice here. The essential ingredient behind a lattice
is that we discretize space, and thus we transition from
functions $f(x)$ to finite vectors $f(x_j) \rightarrow f_j$\footnote{This
is for the case of a finite number of lattice sites. For an infinite, but
countable, number of lattice sites, functions simply become $f \colon \Lambda \to N$,
where $\Lambda$ is the lattice and $N$ the target space. The only essential
difference then between the continuum and lattice case is the countability of the
domain.}. A generic lattice Hamiltonian is then a finite matrix, $H \to h_{ij}$.
For our purposes, we choose to work with the tight-binding model,
\begin{align}
    \label{eq:tight-binding}
    H = \sum_{i} h_{ii} c_i^\dagger c_i - \sum_{ij} (t_{ij} c_j^\dagger c_i + t_{ij}^* c_i^\dagger c_j),
\end{align}
where $t_{ij}$ is the probability amplitude for a particle to hop from site $i$ to site $j$.
The $c_j^\dagger$'s create spinless fermions. Since we are working with a background AB
flux, we introduce the analog of minimal coupling on the lattice. The holonomy becomes
\begin{align}
    \label{eq:holonomy-lattice}
    W_\gamma(A) = \exp\Big( \ii \sum_{(ij) \in \gamma} \theta_{ij} \Big) = \exp\left( \frac{2\pi \ii}{\Phi_0} \oint_\gamma A \right) = \exp\left( \frac{2\pi \ii}{\Phi_0} \int_0^{2\pi} A_\theta \, \dd\theta \right),
\end{align}
where
\begin{align}
    \label{eq:theta-ij}
    \theta_{ij} = \frac{2\pi}{\Phi_0} \int_{i}^{j} \dd\mathbf{s} \cdot \mathbf{A}(\mathbf{s})
\end{align}
is the Peierls phase \cite{Peierls:1933} associated with the hop from site $i$ to site $j$.
Our Hamiltonian becomes
\begin{align}
    \label{eq:tight-binding-AB}
    H = \sum_{i} h_{ii} c_i^\dagger c_i - \sum_{ij} (t_{ij} e^{\ii\theta_{ij}} c_j^\dagger c_i + t_{ij}^* e^{-\ii\theta_{ij}} c_i^\dagger c_j).
\end{align}
Gauge transformations $\zeta$ manifest as phase redefinitions
on individual sites 
\begin{align}
    c_j \mapsto e^{\ii\zeta_j} c_j,
\end{align}
corresponding to $\mathbf{A} \mapsto \mathbf{A} + \grad \zeta$.
The Hamiltonian (\ref{eq:tight-binding-AB}) is gauge-invariant
if $\theta_{ij}$ transforms as
\begin{align}
    \theta_{ij} \mapsto \theta_{ij} + \zeta_j - \zeta_i.
\end{align}
Neither small nor large gauge transformations change the value of
(\ref{eq:tight-binding-AB}), as they only appear through the factors
like $e^{\ii\theta_{ij}}$.

A current is the response of the Hamiltonian to a
background gauge field, $\mathbf{J} = - \delta H / \delta \mathbf{A}$.
On the lattice with our specific model (\ref{eq:tight-binding-AB}), we define the current as
\begin{align}
    \label{eq:tight-binding-current}
    I_{ij} = \frac{2\pi \ii}{\Phi_0} (t_{ij} e^{\ii\theta_{ij}} c_j^\dagger c_i - t_{ij}^* e^{-\ii\theta_{ij}} c_i^\dagger c_j).
\end{align}
$I_{ij}$ is physically interpreted as the current going from site $i$ to site $j$. 
For real $t_{ij} = t$, we have
\begin{align}
    \label{eq:expected-I-ij}
    \langle I_{ij} \rangle = - \frac{4\pi t}{\Phi_0} \operatorname{Im}\left( e^{\ii \theta_{ij}} \langle c_j^\dagger c_i \rangle \right).
\end{align}

A lattice annulus is obtained by taking a disk, discretizing
it, and then excising all lattice sites within a radius $a$
from the origin. Similar to the continuum case, we can
show that the flux values $\nu = 0$ and $\nu = 1/2$ are
privileged in some sense. Specifically, the system will choose
these values so long as time-reversal symmetry is not spontaneously broken.
Via numerics, we can show that this holds for various different lattice
regularizations.

\begin{theorem}
    Let $H$ be the lattice Hamiltonian on the annulus and
    $\nu$ the reduced AB flux. Then, the chosen values
    of $\nu$ are $0$ or $1/2$ as long as time-reversal symmetry
    is not spontaneously broken.
    \label{thm:lattice-flux-trapping}
\end{theorem}

\begin{proof}
    Let $\varphi_j = \arg(x_j + \ii y_j) \in (-\pi, \pi]$ be the polar
    angle of the $j^{\rm th}$ site. Since $\varphi$ has a branch cut, the raw
    difference $\varphi_j - \varphi_i$ is not necessarily the angular difference
    between site $i$ and site $j$. So, we instead define
    \begin{align}
        \label{eq:wrapped-angle}
        & \Delta_{ij} \coloneqq \operatorname{Arg}\left(e^{\ii(\varphi_j - \varphi_i)}\right)
        \in (-\pi, \pi],
        & \theta_{ij}(\nu) = \nu \Delta_{ij}.
    \end{align}
    Physically, $\Delta_{ij}$ is the change in polar angle obtained from
    continuously following the angle along the bond. For a closed cycle $\gamma$,
    \begin{align}
        \sum_{(i,j) \in \gamma} \Delta_{ij} = 2\pi w(\gamma),
    \end{align}
    where $w(\gamma) \in \mathbb{Z}$ is the winding number of $\gamma$ around the
    AB solenoid. This says that contractible cycles carry zero flux, while a cycle that
    winds once around the puncture has holonomy $2\pi\nu$.
    Define the unitary operator
    \begin{align}
        \label{eq:unitary-varphi-j}
        U = \exp\left( \ii \sum_{j} \varphi_j n_j \right),
    \end{align}
    where $n_j = c_j^\dagger c_j$ is the number operator, and $\sum_j$ sums over 
    all lattice sites. It is simple to check that
    \begin{align}
        & U c_j U^{-1} = e^{-\ii \varphi_j} c_j
        & U c_j^\dagger U^{-1} = e^{\ii \varphi_j} c_j^\dagger.
    \end{align}
    Although the real-valued angle $\varphi_j$ jumps across its branch cut,
    $e^{\ii\varphi_j}$ is single-valued on the lattice sites. Furthermore,
    \begin{align}
        e^{\ii(\varphi_j - \varphi_i)} = e^{\ii\Delta_{ij}},
    \end{align}
    and so conjugating a hopping term by $U$ changes its Peierls factor from
    $e^{\ii\nu\Delta_{ij}}$ to $e^{\ii(\nu+1)\Delta_{ij}}$.
    Conjugating (\ref{eq:tight-binding-AB}) by (\ref{eq:unitary-varphi-j})
    gives
    \begin{align}
        U H(\nu) U^{-1} = H(\nu + 1),
    \end{align}
    and thus $\spec(H(\nu)) = \spec(H(\nu + 1))$.
    Applying time-reversal symmetry to (\ref{eq:tight-binding-AB}) gives
    \begin{align}
        \label{eq:lattice-even}
        \mathcal{T} H(\nu) \mathcal{T}^{-1} = H(-\nu),
    \end{align}
    and thus the spectrum of $H$ is even in $\nu$. This follows from
    the fact that the only $\nu$-dependence in (\ref{eq:tight-binding-AB}) 
    is in the factors $\theta_{ij}$. By the same arguments in
    the proof of Theorem \ref{thm:0-or-1/2}, $\nu \in \{ 0, 1/2 \}$
    unless time-reversal symmetry is spontaneously broken.
\end{proof}

Since there is no magnetic field outside the AB solenoid,
we can choose a vector potential that is pure gauge. Define $\varphi$
to be the polar angle. Then, we can choose the symmetric gauge for
$\mathbf{A}(\bfx)$,
\begin{align}
    & \mathbf{A}(\bfx) = \frac{\Phi}{2\pi} \grad \varphi,
    & \oint \dd\bfx \cdot \mathbf{A} = \Phi.
\end{align}
The Peierls phase becomes
\begin{align}
    \label{eq:peierls-phase}
    \theta_{ij}(\nu)
    = \frac{2\pi}{\Phi_0} \int_i^j \dd\mathbf{s} \cdot \mathbf{A}(\mathbf{s})
    = \nu \Delta_{ij},
\end{align}
with $\Delta_{ij}$ defined in (\ref{eq:wrapped-angle}). This is
the phase we use in the lattice numerics below.

Now we numerically investigate flux trapping on the lattice.
We take a square-lattice annulus of outer radius $b = 16$ in units
of the lattice spacing, with hoppings $t_{ij} = t = 1$. We take the sites
to have half-integer coordinates so that the flux pierces the center of
a plaquette. We use five different choices of UV data at the puncture
site:
\begin{table}[H]
    \centering
    \begin{tabular}{|l|c|l|c|}
        \hline
        Regularization & Sites removed & Edge modification & Sites \\
        \hline
        Plaquette puncture & $0$ & None & $812$ \\
        Four-site hole & $4$ & None & $808$ \\
        Edge barrier & $4$ & $V = 1.5\,t$ on the $8$ central sites & $808$ \\
        Weak edge hopping & $4$ & $t' = 0.45\,t$ on central edge bonds & $808$ \\
        Larger hole & $16$ & None & $796$ \\
        \hline
    \end{tabular}
    \caption{The five different UV regularizations used in the scans
    below. ``Central sites'' refers to the eight sites nearest the central plaquette.
    Adding a potential $V$ to these sites means that we add a term $\sum_{\text{central}} V_i c_i^\dagger c_i$
    to the Hamiltonian (\ref{eq:tight-binding-AB}).}
    \label{tab:lattice-regulators}
\end{table}
For each choice, we diagonalize the Hamiltonian (\ref{eq:tight-binding-AB}) 
on the lattice for the fundamental domain of $0 \leq \nu \leq 1/2$.
We then compute the negative of the grand canonical potential
\begin{align}
    \Omega_\beta(\nu, \mu) = - \Omega_{\rm gc} = \beta^{-1} \sum_\alpha \log\left( 1 + e^{-\beta(E_\alpha(\nu) - \mu)}\right),
\end{align}
where $E_\alpha(\nu)$ is the energy of eigenstate $\alpha$ and $\mu$ the
chemical potential. The system will choose $\nu$ such that $\Omega_\beta(\nu, \mu)$
is maximized. We sample a range of $\beta$ and $\mu$ values---the
full sweeps are listed in Table~\ref{table:sweeps}. As a check on the
flux-periodicity of the spectrum, $\spec H(\nu + 1) = \spec H(\nu)$, 
the spectra of $H(1/4)$ and $H(5/4)$ agree to $2 \cdot 10^{-14} \, t$.
Evenness needn't be checked, as $H(-\nu) = H(\nu)^*$, so the equality 
of the spectra is exact in floating-point arithmetic.

\begin{table}[H]
    \centering
    \begin{tabular}{|c|c|}
        \hline
        $\beta t$ & $\{1,\,10,\,20,\,50,\,70,\,90,\,200\}$ and $T=0$ \\
        \hline
        $\mu/t$ & $\{-3.2,\,-2,\,-1,\,0,\,1\}$ \\
        \hline
    \end{tabular}
    \caption{Parameter sweeps used for the lattice flux-trapping scans.}
    \label{table:sweeps}
\end{table}

Scanning over these parameters---along with the five regularizations from
earlier---gives a total of $200$ different simulations. These simulations confirm that
the flux values $\nu = 0$ and $\nu = 1/2$ are trapped every time. A few 
representative graphs are shown below:

\begin{figure}[H]
    \centering
    \includegraphics[width=0.95\textwidth]{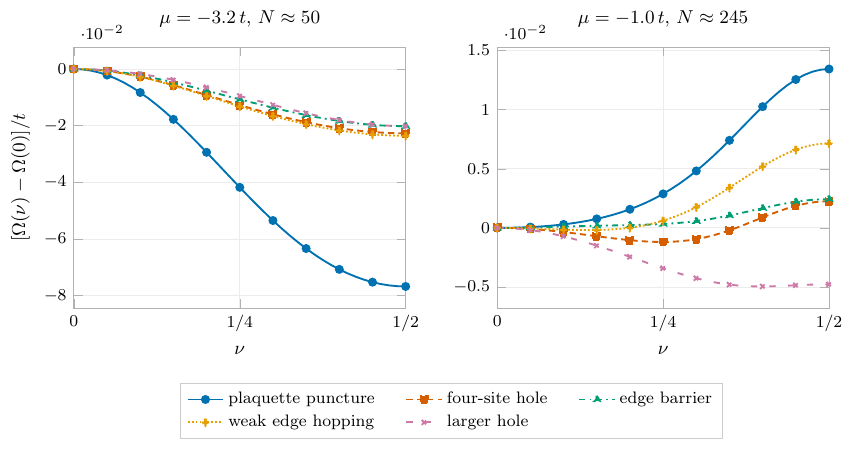}
    \caption{Representative lattice flux-trapping landscapes, plotted as
    $(\Omega(\nu) - \Omega(0))$ at $\beta t = 90$ for the five puncture
    regularizations. Recall that the maximum of $(\Omega(\nu) - \Omega(0))$
    dictates where the flux will be trapped. Left: at low filling, $\mu = - 3.2 t$
    and $N \approx 50$, flux is trapped at $\nu = 0$. Right: at higher
    filling, $\mu = - t$ and $N \approx 245$, the larger hole traps
    $\nu = 0$ while the other regularizations trap $\nu = 1/2$.}
    \label{fig:lattice-flux-landscapes}
\end{figure}

\section*{Acknowledgements}
\label{sec:acknowledgements}

We thank F. Popov and Z. Komargodski for their support, suggestions, and for pointing
out errors in earlier drafts of this paper. We also thank Claude and Codex for useful discussions
and for their help in generating the code and figures in this paper. 
The author gratefully acknowledges financial support from the Simons Center.

\appendix

\section{\texorpdfstring{Finite-$R$}{Finite-R} Current Calculation}

\label{app:finite-R}
Recall that for a sufficiently smooth function $f$ that takes values 
on the lattice $\{ x_n = a_0 + n a \}_{n \in \ints}$, with $a$ the 
lattice spacing, we have
\begin{align}
    \label{eq:euler-maclaurin}
    a \sum_{n = 0}^N f(a_0 + n a) & = \int_{a_0}^{a_0 + Na} \dd x \, f(x) + \frac{a}{2} \left( f(a_0) + f(a_0 + Na) \right)\\
    & + \frac{a^2}{12}\left(f'(a_0+Na)-f'(a_0)\right)+O(a^4).
\end{align}
For $r \gg k_F^{-1}$, the Bessel zeros are approximately
\begin{align}
    \label{eq:approx-bessel-zeros}
    j_{|\ell|, p} \simeq \jelp^{(0)} = \pi \left( p + \frac{|\ell|}{2} - \frac{1}{4} \right) + O(\ell^2/R).
\end{align}
Define $k_{|\ell|, p} = j_{|\ell|, p} / R$. Then $k_{|\ell|, p + 1} - k_{|\ell|, p} \simeq \pi / R$
at fixed $|\ell|$. So, $a = \pi / R$ is our lattice spacing. 
Near the Fermi surface, where $j_{|\ell|, p} \simeq R$,
the corrections to the Bessel zeros (\ref{eq:approx-bessel-zeros})
are of order $O(\ell^2/R)$. Since only channels with $|\ell| \lesssim r$
contribute at radius $r$, the range of validity of this approximation
is $r \ll R^{1/2}$. Using this approximation requires another condition
on the Bessel zeros. The radial filling condition is
\begin{align}
    \jelp < R \iff p < \frac{1}{a} - \frac{|\ell|}{2} + \frac{1}{4}.
\end{align}
This defines a maximum integer $p$, which we'll write as
\begin{align}
    P(|\ell|) = \left\lfloor \frac{1}{a} - \frac{|\ell|}{2} + \frac{1}{4} \right\rfloor,
\end{align}
in the case that the argument is not an integer. If the
argument is an integer, then we instead use $P(|\ell|) = \lceil \cdots \rceil - 1$. $P(|\ell|)$ is the upper bound
on the radial sum over the Bessel zeros.
The first correction to the Bessel zeros, $\jelp^{(1)} = O(\ell^2/R)$,
must obey
\begin{align}
    \label{eq:non-resonance-condition}
    & j_{|\ell|, P(|\ell|)}^{(0)} + j_{|\ell|, P(|\ell|)}^{(1)} < R,
    & j_{|\ell|, P(|\ell|) + 1}^{(0)} + j_{|\ell|, P(|\ell|) + 1}^{(1)} \geq R.
\end{align}
This is to ensure that we are not filling the wrong number of
radial levels in the sum. We will say that this is a ``non-resonance''
assumption, and it is actually worth discussing in a bit more detail.

The first-order correction to the Bessel zero is exactly given by
\begin{align}
    \jelp^{(1)} = -\frac{4|\ell|^2-1}{8j_{|\ell|,p}^{(0)}}.
\end{align}
Near the Fermi surface, the approximate and exact filling conditions
disagree when the inequality (\ref{eq:non-resonance-condition}) is violated,
that is when
\begin{align}
    |R - \jelp^{(0)}| \lesssim |\jelp^{(1)}| = O(|\ell|^2/R).
\end{align}
This says that for each angular-momentum channel, there is a small interval
between the approximate and exact Bessel zeros within which the two filling prescriptions
disagree---that is, the integer $P(|\ell|)$ calculated from the approximate
filling condition differs from the integer $P(|\ell|)$ calculated from the exact
filling condition.

Since the inequalities (\ref{eq:non-resonance-condition}) are restrictions on the size
of the first-order correction to the Bessel zero, they're really conditions on the
size of $R$ in relation to the angular momentum $|\ell|$. So, the violation of
(\ref{eq:non-resonance-condition}) defines a set of ``bad'' outer radii, $\mathcal{R}_{\rm bad}(|\ell|)$,
for which the approximate and exact filling conditions disagree. Hopefully, this 
set of bad radii should be small in comparison to the total range of outer radii
as indexed by $|\ell|$---this turns out to indeed be the case. We now quantify this statement.

Fix an angular momentum channel $\alpha$. Then we ask:
what fraction of the radii at fixed $\alpha$ are ``bad''?
The condition for a radius to be bad is that adding the
first-order correction to the Bessel zero changes the
filling level (\ref{eq:non-resonance-condition}). The fraction
of bad radii depends on the interval we sample from.
Say we can choose radii in the interval $(0, L]$.
Define the set of bad radii to be
\begin{align}
    \mathcal{R}_{\rm bad}(\alpha) = \begin{cases}
        \bigcup_{p \geq 1} (j_{\alpha, p}^{(0)}, j_{\alpha, p}] & 0 \leq \alpha < 1/2\\
        \varnothing & \alpha = 1/2\\
        \bigcup_{p \geq 1} (j_{\alpha, p}, j_{\alpha, p}^{(0)}] & \alpha > 1/2.
    \end{cases}
\end{align}
If the outer radius of our system is in this set,
the filling condition changes when we add the first-order
correction to the Bessel zero. The size of this set relative
to the interval $(0, L]$ is
\begin{align}
    \frac{|\mathcal{R}_{\rm bad}(\alpha) \cap (0, L]|}{L} = O\left( \frac{\alpha^2 \log{L}}{L} \right).
\end{align}
The current, however, involves a sum over all angular-momentum
channels with $|\ell|\lesssim r$. Define the corresponding bad-radius
set by
\begin{align}
    \mathcal R_{\rm bad}^{\rm ang}(r,\nu) = 
    \bigcup_{\substack{\ell \in \ints - \nu\\|\ell| \lesssim r}} \mathcal R_{\rm bad}(|\ell|).
\end{align}
The individual bad sets are strongly correlated with
one another. This is because
\begin{align}
    j_{\alpha + 2, p - 1}^{(0)} = j_{\alpha, p}^{(0)}.
\end{align}
So, the angular momenta split into four subfamilies:
\begin{align}
    \alpha = \nu + 2k,
    \quad 
    1 + \nu + 2k,
    \quad
    1 - \nu + 2k, 
    \quad 
    2 - \nu + 2k.
\end{align}
Each subfamily is controlled by its largest angular momentum below $r$,
rather than by the sum of all its channels. For fixed $r$,
\begin{align}
    \frac{|\mathcal{R}_{\rm bad}^{\rm ang}(r, \nu) \cap (0, L]|}{L}
    =
    O\left( \frac{r^2 \log L}{L} \right).
\end{align}
Equivalently, in a local interval near a large outer radius $R$, the
fraction of radii for which at least one contributing channel is
miscounted is
\begin{align}
    O\left(\frac{r^2}{R}\right).
\end{align}
We can see that this fraction vanishes throughout the strict regime
$r \ll R^{1/2}$.

At a bad radius, the incorrect filling condition
can change the coefficient of the finite-$R$ correction by 
an order-one amount. It does not, however, change the scaling
of the correction---linear scaling is retained even if a radius is
bad. So, the coefficient is sensitive to bad radii, but the linear scaling
of the current correction and thus the scaling of the crossover radius, $r_*$, 
remain unchanged.

For example, take $R = 500$, $\nu = 0.3$, and
$r = R^{1/3} \simeq 7.94$. The estimated fraction
of bad radii is approximately $6.2\%$, while numerics
give $5.9\%$. So, roughly $94\%$ of nearby outer radii satisfy the
non-resonance condition for every angular-momentum channel
contributing at $r$. In the bad radius cases, the exact current
differs from the predicted current by an average of $166\%$
(this is due to switching from a small positive value to a small
negative value), and flux trapping is retained in all cases.

Using the approximations above, the finite-$R$ $Q$-term (\ref{eq:Q-r}) then becomes
\begin{align}
    \label{eq:Q-R-approx}
    Q_R(r) \simeq \frac{a}{2\pi} \sum_{\ell \in \ints - \nu} \ell \, \sum_p k_{|\ell|, p} J_{|\ell|}(k_p r)^2,
\end{align}
with $a = \pi / R$ the lattice spacing from before. Using (\ref{eq:euler-maclaurin})
gives
\begin{align}
    \label{eq:delta-Q}
    \Delta Q(r, \nu) = \frac{1}{2R} \sum_{\ell \in \ints - \nu} \ell \left( \frac{1}{2} - \left\{ \frac{R}{\pi} - \frac{|\ell|}{2} + \frac{1}{4} \right\}\right) J_{|\ell|}(r)^2 + \mathcal{R}_{\rm rad},
\end{align}
where $\{ \cdots \} = (\cdots) - \lfloor \cdots \rfloor$ denotes the
fractional part of $(\cdots)$ and
\begin{align}
    \mathcal{R}_{\rm rad} = \sum_{\ell\in\ints-\nu} \ell \, \mathcal E_{|\ell|,R}
\end{align}
is the sum of the leftover Euler--Maclaurin terms. Specifically,
\begin{align}
    \mathcal{E}_{\alpha, R} & = \frac{1}{\pi R^2} \sum_{j_{\alpha, p} < R} \left( \frac{J_\alpha(j_{\alpha, p} (r/R))}{J_{\alpha + 1}(j_{\alpha, p})} \right)^2\\
    & - \frac{1}{2\pi} \int_0^1 \dd k \, k J_{\alpha}(kr)^2\\
    & - \frac{1}{2R} \left( \frac{1}{2} - \left\{ \frac{R}{\pi} - \frac{\alpha}{2} + \frac{1}{4} \right\}\right) J_\alpha(r)^2.
\end{align}
Define
\begin{align}
    \label{eq:sigma-def}
    \sigma_R(\alpha) = \frac{1}{2} - \left\{ \frac{R}{\pi} - \frac{\alpha}{2} + \frac{1}{4} \right\}
\end{align}
to rewrite (\ref{eq:delta-Q}) as
\begin{align}
    \label{eq:new-delta-q-sum}
    \Delta Q(r, \nu) = \frac{1}{2R} \sum_{\ell \in \ints - \nu} \ell \, \sigma_R(|\ell|) J_{|\ell|}(r)^2 + \mathcal{R}_{\rm rad}.
\end{align}
We now make a number of approximations to be able to
explicitly perform the sum (\ref{eq:new-delta-q-sum}).
\begin{enumerate}
    \item Replace $J_\alpha(r)^2$ with its Debye envelope,
    \begin{align}
        \label{eq:j-debye}
        J_\alpha(r)^2 = \frac{1}{\pi \sqrt{r^2 - \alpha^2}} + \text{oscillatory terms}.
    \end{align}
    The oscillatory terms, when summed over $\ell$, give us an error term
    that we call $\mathcal{R}_{\rm osc}$. This is not analytically bounded,
    but numerical evidence shows that this term scales as $o(r/R)$.
    \item Replace the upper bound on the sum with $M = r - r^{1/3}$
    instead of $\infty$. The Debye approximation above is not valid
    when $\alpha$ lies within $O(r^{1/3})$ of $r$, as the denominator
    diverges when $\alpha \to r$. The omitted region contains $O(r^{1/3})$
    angular momentum channels, and within this region, $|\ell| = O(r)$.
    So, $\sigma_R(|\ell|) \leq 1/2$, and $J_{|\ell|}(r) = O(r^{-1/3})$.
    This implies that each omitted term is $O(r^{1/3}/R)$, and summing over
    the $O(r^{1/3})$ channels gives an error term $\mathcal{R}_{\rm turn}
    = O(r^{2/3}/R)$.
    \item It is useful to pair the $\ell = m + 1 - \nu$ and $\ell = -m - \nu$
    terms together. Doing this and using the approximations above gives the
    sum 
    \begin{align}
        \label{eq:updated-sum}
        \Delta Q(r, \nu) & = \frac{1}{2R} \sum_{m < M} \left( \sigma_R(m + 1 - \nu) \frac{m + 1 - \nu}{\pi\sqrt{r^2 - (m + 1 - \nu)^2}} \right.\\
        & \left. - \sigma_R(m + \nu) \frac{m + \nu}{\pi \sqrt{r^2 - (m + \nu)^2}} \right) + \mathcal{R}_{\rm osc} + \mathcal{R}_{\rm turn} + \mathcal{R}_{\rm rad}.
    \end{align}
    The leading behavior of the error terms comes
    from $\mathcal{R}_{\rm turn} = O(r^{2/3}/R)$.
    Define the weight function
    \begin{align}
        \label{eq:weight}
        w_r(\alpha) = \frac{\alpha}{\pi\sqrt{r^2 - \alpha^2}}.
    \end{align}
    We can then approximate the weight terms in (\ref{eq:updated-sum})
    as $w_r(m + 1 - \nu) \approx w_r(m + \nu)$, picking up an $O(r^{1/3}/R)$
    error term $\mathcal{R}_{\rm shift}$. This gives the sum
    \begin{align}
        \label{eq:sum-same-weights}
        \Delta Q(r, \nu) = \frac{1}{2R} \sum_{m < M} (\sigma_R(m + 1 - \nu) - \sigma_R(m + \nu)) w_r(m + \nu) + O(r^{2/3}/R) + o(r/R).
    \end{align}
    \item We now note that the sequence
    \begin{align}
        \zeta_R(m) = \sigma_R(m + 1 - \nu) - \sigma_R(m + \nu)
    \end{align}
    has period $2$, that is, $\zeta(m + 2) = \zeta(m)$. Recall that
    any period-two sequence can be split into a constant term plus
    an alternating term. We write this as
    \begin{align}
        \zeta_R(m) = c_R + (-1)^m a_R,
    \end{align}
    with
    \begin{align}
        & c_R = \frac{\zeta(0) + \zeta(1)}{2}
        & a_R = \frac{\zeta(0) - \zeta(1)}{2}.
    \end{align}
    Note that these constants also depend on $\nu$.
    It is useful to explicitly write the definition of $c_R(\nu)$.
    First, define $t = R/\pi + 1/4$. Then,
    \begin{align}
        \label{eq:cr}
        c_R(\nu) & = \frac{1}{2}\left(
            \sigma_R(1 - \nu) - \sigma_R(\nu)
            + \sigma_R(2 - \nu) - \sigma_R(1 + \nu)
        \right)\\
        & = \frac{1}{2}\left(
            \left\{ t - \frac{\nu}{2} \right\} - \left\{ t - \frac{1}{2} + \frac{\nu}{2} \right\}
            + \left\{ t - \frac{1}{2} - \frac{\nu}{2} \right\} - \left\{ t + \frac{\nu}{2} \right\}
        \right).
    \end{align} 
    Now, this division of $\zeta(\nu; m)$ splits the sum (\ref{eq:sum-same-weights}) 
    into a constant and an alternating part. The alternating terms mostly cancel against each other, 
    leaving an $O(r^{1/3})$ error. One can see this by pairing neighboring terms in the alternating
    sum together and noting that since $w_r(m)$ is increasing as a function
    of $m$, the sum of these differences is bounded by the highest term,
    $O(w_r(M)) = O(r^{1/3})$. Call the error from this term 
    $\mathcal{R}_{\rm alt}$, which scales as $\mathcal{R}_{\rm alt} = O(r^{1/3}/R)$.
    \item Finally, approximate the sum by an integral
    \begin{align}
        \label{eq:sum-to-int}
        \sum_{m < M} w_r(m) = \int_0^M \dd x \, w_r(x) + O(r^{1/3}).
    \end{align}
    Evaluating this integral gives
    \begin{align}
        \label{eq:int-answer}
        \int_0^M \dd x \, w_r(x) = \frac{1}{\pi} (r - \sqrt{r^2 - M^2}),
    \end{align}
    and with $M^2 = (r - r^{1/3})^2$, (\ref{eq:int-answer})
    becomes
    \begin{align}
        \frac{1}{\pi}\left(r - \sqrt{r^2 - M^2}\right) = \frac{1}{\pi} \left(r - O(r^{2/3})\right).
    \end{align}
\end{enumerate}
This leaves a final answer of
\begin{align}
    \label{eq:final-sum}
    \Delta Q(r, \nu) = \frac{c_R}{2\pi} \frac{r}{R} + O\left( \frac{r^{2/3}}{R} \right) + o(r/R),
\end{align}
where the scaling $O(r^{2/3}/R)$ is analytically determined
and the scaling $o(r/R)$ is numerically supported.

\section{\texorpdfstring{Infinite-$R$}{Infinite-R} Current Calculation}
\label{app:infinite-R}

We aim to calculate the infinite-disk current (\ref{eq:Q-infty}).
Substituting $u = kr$ gives
\begin{align}
    Q_\infty(r, \nu) = \frac{1}{2\pi r^2} \sum_{\ell \in \ints - \nu} \ell \int_0^r \dd u \, u \, J_{|\ell|}^2(u).
\end{align}
Define the series
\begin{align}
    \mathcal{J}_\nu(u) \coloneqq \sum_{n = 0}^\infty (n + \nu) J_{n + \nu}^2(u).
\end{align}
Splitting the sum gives
\begin{align}
    Q_\infty(r, \nu)
    = \frac{1}{2\pi r^2} \sum_{\ell \in \ints - \nu} \ell \int_0^r \dd u \, u \, J_{|\ell|}^2(u)
    = \frac{1}{2\pi r^2} \int_0^r \dd u \, u \left(\mathcal{J}_{1 - \nu}(u) - \mathcal{J}_{\nu}(u)\right).
\end{align}
Taking the derivative of $\mathcal{J}_\nu(u)$ and using $2J_\nu' = J_{\nu-1} - J_{\nu+1}$
and $(2\nu/u)J_\nu = J_{\nu-1} + J_{\nu+1}$ produces
\begin{align}
    \mathcal{J}_\nu'(u) = 2 \sum_{n \geq 0} (n + \nu) J_{n + \nu} J_{n + \nu}' = \frac{u}{2} \sum_{n \geq 0} (J_{n + \nu - 1}^2 - J_{n + \nu + 1}^2) = \frac{u}{2} \left( J_{\nu - 1}^2 + J_{\nu}^2 \right),
\end{align}
where we used the fact that the series is telescoping in the last step.
We may use the Lommel integral to integrate this expression, giving
\begin{align}
    I_\alpha(u) \coloneqq \int_0^u \dd v \, v \, J_\alpha^2(v)
    = \frac{u^2}{2} \left( J_\alpha^2 - J_{\alpha - 1}J_{\alpha + 1} \right),
\end{align}
and so
\begin{align}
    \mathcal{J}_\nu(u) = \frac{1}{2} \left( I_{\nu - 1}(u) + I_{\nu}(u) \right) + C.
\end{align}
The integration constant is fixed as $C = 0$ by $\mathcal{J}_\nu(0) = 0$.
Substituting this into our expression for $Q_\infty(r, \nu)$ gives
\begin{align}
    Q_\infty(r, \nu) = \frac{1}{4\pi r^2} \int_0^r \dd u \, u \left( I_{-\nu}(u) + I_{1 - \nu}(u) - I_{\nu - 1}(u) - I_{\nu}(u) \right).
\end{align}
Thus our task is to evaluate expressions in the form of
\begin{align}
    \mathcal{I}_\alpha(r) \coloneqq \int_0^r \dd u \, u \, I_{\alpha}(u) = \frac{1}{2} \int_0^r \dd u \, u^3 (J_\alpha^2 - J_{\alpha - 1} J_{\alpha + 1}).
\end{align}
We can rewrite the second term using $J_{\alpha \pm 1} = \mp J_\alpha' 
+ (\alpha/u) J_{\alpha}$. Then,
\begin{align}
    \mathcal{I}_\alpha(r) = \frac{1}{2} \int_0^r \dd u \, u^3 \left[J_{\alpha}'^2 + J_{\alpha}^2 \left( 1 - \frac{\alpha^2}{u^2} \right) \right].
\end{align}
We now use Bessel's equation,
\begin{align}
    J_\alpha'' + \frac{1}{u}J_\alpha' + \left( 1 - \frac{\alpha^2}{u^2} \right) J_\alpha = 0.
\end{align}
One can check that
\begin{align}
    \mathcal{F}_\alpha(u) \coloneqq
        &\left(\frac{u^4}{6} + \frac{\alpha^2 u^2}{6} + \frac{\alpha^2(1 - \alpha^2)}{3}\right) J_\alpha^2
        - \frac{u^3}{3} J_\alpha J_\alpha'
        + \left(\frac{u^4}{6} + \frac{(\alpha^2 - 1)u^2}{3}\right) J_\alpha'^2
\end{align}
satisfies
\begin{align}
    \frac{\dd}{\dd u} \mathcal{F}_\alpha(u) = u^3 J_\alpha^2.
\end{align}
For positive $\alpha$, the lower endpoint clearly vanishes. The dangerous cases
are $\alpha = \nu - 1$ and $\alpha = -\nu$, since both are in $(-1, 0)$.
In these cases,
\begin{align}
    & J_\alpha(u) \sim c_\alpha u^\alpha,
    & J_\alpha'(u) \sim \alpha c_\alpha u^{\alpha - 1}.
\end{align}
The leading terms in $\mathcal{F}_\alpha(u)$ are then
\begin{align}
    \frac{\alpha^2(1 - \alpha^2)}{3} c_\alpha^2 u^{2\alpha}
    + \frac{\alpha^2(\alpha^2 - 1)}{3} c_\alpha^2 u^{2\alpha} = 0,
\end{align}
so the singular terms actually cancel. Thus
$\mathcal{F}_{\nu - 1}(u) = O(u^{2\nu})$ and
$\mathcal{F}_{-\nu}(u) = O(u^{2 - 2\nu})$, both of which vanish at the lower
endpoint because $0 < \nu < 1$.
Then,
\begin{align}
    \int_0^r \dd u \, u^3 J_\alpha^2
    = &\left(\frac{r^4}{6} + \frac{\alpha^2 r^2}{6} + \frac{\alpha^2(1 - \alpha^2)}{3}\right) J_\alpha^2
        - \frac{r^3}{3} J_\alpha J_\alpha'\\
    &\quad+ \left(\frac{r^4}{6} + \frac{(\alpha^2 - 1)r^2}{3}\right) J_\alpha'^2.
\end{align}
Equivalently, integrating by parts,
\begin{align}
    \mathcal{I}_\alpha(r) = \frac{r^2}{2} I_\alpha(r) - \frac{1}{2} \int_0^r \dd u \, u^3 J_\alpha^2.
\end{align}
Substituting the expression above into this identity gives the finite boundary
value
\begin{align}
    \mathcal{I}_\alpha(r)
    = \frac{1}{6} \left[
        \left(r^4 - 2\alpha^2 r^2 + \alpha^2(\alpha^2 - 1)\right)J_\alpha^2
        + r^3 J_\alpha J_\alpha'
        + r^2(r^2 - \alpha^2 + 1) J_\alpha'^2
    \right].
\end{align}
Here all of the Bessel functions are evaluated at $r$.
Finally, the current may be written as
\begin{align}
    \label{eq:current_in_terms_of_I}
    Q_\infty(r, \nu) = \frac{1}{4\pi r^2} \left( \mathcal{I}_{-\nu}(r) + \mathcal{I}_{1 - \nu}(r) - \mathcal{I}_{\nu - 1}(r) - \mathcal{I}_{\nu}(r) \right).
\end{align}
Define $S_\nu(r)$ as in \cite{Komargodski:2026vwn}:
\begin{align}
    S_\nu(r) & \coloneqq \left( \frac{2r^2 - 2\nu^2 + 5\nu}{6} \right) J_\nu^2
    + \left( \frac{2r^2 - 2\nu^2 + 3\nu - 1}{6} \right) J_{\nu + 1}^2\\
    &+ \left( \frac{\nu(2\nu - 1)(\nu - 1)}{3r} - \frac{(2\nu + 1)r}{3} \right) J_\nu J_{\nu + 1}.
\end{align}
We'd like to relate this to (\ref{eq:current_in_terms_of_I}).
Recalling that
\begin{align}
    & J_\nu' = \frac{\nu}{r}J_\nu - J_{\nu + 1},
    & J_{\nu - 1} = \frac{2\nu}{r}J_\nu - J_{\nu + 1},
\end{align}
and using the definition of $\mathcal{I}_\alpha(r)$ gives
\begin{align}
    \frac{1}{2}\left(\mathcal{I}_{\nu - 1}(r) + \mathcal{I}_{\nu}(r)\right)
    = \frac{r^2}{2} S_\nu(r).
\end{align}
Thus our current becomes
\begin{align}
    Q_\infty(r, \nu)
    &= \frac{1}{2\pi r^2} \left[
        \frac{1}{2} \left( \mathcal{I}_{-\nu}(r) + \mathcal{I}_{1 - \nu}(r) \right)
        - \frac{1}{2} \left( \mathcal{I}_{\nu - 1}(r) + \mathcal{I}_{\nu}(r) \right)
    \right] \\
    &= \frac{1}{4\pi} \left( S_{1 - \nu}(r) - S_\nu(r) \right).
\end{align}

Since $Q_\infty(r, \nu) = r^2 j_\infty^\theta(r)$, raising the index with
the cylindrical metric gives the current density
\begin{align}
    \label{eq:final_j_upper_theta}
    j_\infty^\theta(r) = \frac{1}{4\pi r^2} \left( S_{1 - \nu}(r) - S_\nu(r) \right).
\end{align}
At distances $rk_F \gg 1$, we can use the large-$r$ approximation for
the Bessel functions. In this regime, we find that the current goes like
\begin{align}
    j_\infty^\theta(r) = \frac{(2\nu - 1)\sin(\pi\nu) \sin(2r)}{8\pi^2 r^4} + O(r^{-5}),
\end{align}
and thus
\begin{align}
    Q_\infty(r, \nu) = \frac{(2\nu - 1) \sin(\pi \nu) \sin(2r)}{8\pi^2 r^2} + O(r^{-3}),
\end{align}
which is (\ref{eq:q-infty-exact}).

\section{UV Near-Origin Current}
\label{app:uv-current}

In this appendix, we show that there is a persistent
current localized near the origin in the infinite-volume
limit of the disk, at both zero and finite temperature.

Note that the above analyses necessarily miss this effect,
as those were IR analyses, and this is a UV effect. The current
is localized to
\begin{align}
    \label{eq:current-localization}
    r \sim k_F^{-1} e^{-1/|\nu|},
\end{align}
which follows from the near-origin behavior of 
the Bessel function in (\ref{eq:inf-vol-current}).
This is outside the range of the effective theory of the Fermi surface,
and thus the current is non-universal.

\subsection{Zero-Temperature Disk}
\label{app:uv-zero-T}

We now consider the case of the infinite disk, i.e. the
$a \to 0$ and $b \to \infty$ limit of an annulus of inner radius
$a$ and outer radius $b$. In this limit, the wave functions become
\begin{align}
    \label{eq:inf-disk-wave-functions}
    \psi_{\ell, k} = \sqrt{\frac{k}{2b}} e^{\ii\ell \theta} J_{|\ell|}(kr),
\end{align}
with $k^2 = E$. Calculating the current from these wave functions
using $j^\theta = (1/r^2) \operatorname{Im}(\psi^* \partial_\theta \psi)$ gives
\begin{align}
    \label{eq:inf-vol-current}
    j^\theta = \sum_{\ell \in \ints - \nu} \sum_{E < 1} \frac{k\ell}{2br^2} (J_{|\ell|}(kr))^2.
\end{align}
Computing the net current through the disk gives
\begin{align}
    \label{eq:anomalous-current-int}
    I_{\rm tot} = \int_0^\infty r \, \dd r \, j^\theta = \frac{1}{4b} \sum_{\ell \in \ints - \nu} \int_0^1 \dd E \, k \frac{\dd N_\ell}{\dd E} \sgn(\ell).
\end{align}
In the large-volume limit, the density of states is
\begin{align}
    \frac{\dd N_\ell}{\dd E} = \frac{b}{2\pi k},
\end{align}
and thus,
\begin{align}
    \label{eq:tot-current}
    I_{\rm tot} = \frac{1}{8\pi} \sum_{\ell \in \ints - \nu} \sgn(\ell) = \frac{1}{8\pi} \sum_{n \in \ints} \sgn(n - \nu).
\end{align}
This sum is generically infinite and must be regularized. Before
doing this, let us take a second to recognize the fact that $I_{\rm tot}$
is the $\eta$-invariant of Atiyah, Patodi, and Singer \cite{Atiyah:1975jf}. Consider the operator
\begin{align}
    D(\nu) = p_\theta - \nu = - \ii \partial_\theta - \nu
\end{align}
acting on single-valued periodic functions on $S^1$. This one-dimensional 
operator is first-order and self-adjoint, and thus it's a Dirac operator.
Its spectrum is
\begin{align}
    D(\nu) e^{\ii n\theta} = (n - \nu) e^{\ii n\theta},
\end{align}
and thus its eigenvalues are $\lambda_n = n - \nu$ for $n \in \ints$.
Recall the definition\footnote{We use the standard convention of
$\sgn(0) = 0$. This is also the physically sensible choice in this 
problem, as intuitively, $\eta$ just counts how many left- vs.\ 
right-moving angular momentum modes we have, and $\ell = 0$
corresponds to neither a left- nor right-moving angular momentum mode.} 
of the $\eta$-invariant of a Dirac operator $D$:
\begin{align}
    \eta(D) = \sum_{\lambda_n} \sgn(\lambda_n),
\end{align}
where the sum above is regularized, as $D$ has infinitely
many eigenvalues. For $D(\nu)$, the $\eta$-invariant is given
by
\begin{align}
    \eta(D(\nu)) \eqqcolon \eta_D = \sum_{n \in \ints} \sgn(n - \nu).
\end{align}
This implies that the current may be written as
\begin{align}
    I_{\rm tot} = \frac{1}{8\pi} \eta_D.
\end{align}
This also gives a nice physical interpretation of the $\eta$-invariant:
it simply counts the number of left- and right-moving angular momentum 
modes in the system. Regulating the sum by using an exponential-damping factor,
$e^{-\epsilon|\ell|} = e^{-\epsilon|n - \nu|}$, gives the regulated answer
\begin{align}
    \label{eq:reg-current}
    I_{\rm tot} = \frac{1}{8\pi} (2\nu - 1) = \frac{1}{8\pi} \eta_D.
\end{align}
An interesting fact about this current is that it discontinuously jumps 
from positive to negative as the flux $\Phi$ is turned on. Precisely,
\begin{align}
    \label{eq:current-jump}
    \Delta I_{\rm tot} 
    = \left(\lim_{\nu \to 0^+} I_{\rm tot} \right) 
    - \left(\lim_{\nu \to 0^-} I_{\rm tot} \right) 
    = - \frac{1}{4\pi}.
\end{align}
This is actually obvious from the point of view of the $\eta$-invariant.
From the definition of $\eta_D$, we can see that if $\nu$ crosses an integer value, 
then the value of  $\eta_D$ jumps by $\pm 2$. Thus, crossing $0$ from below produces a 
$-2$  change in the value of $\eta_D$, which implies
\begin{align}
    \Delta I_{\rm tot} = \frac{1}{8\pi} \Delta \eta_D = - \frac{1}{4\pi}.
\end{align}
So, we can see that the current jump from the AB flux crossing zero
from below records the spectral flow of the Dirac operator $D(\nu)$,
with $\SF = \Delta \eta_D / 2$. Thus, we can rewrite the above
as
\begin{align}
    \label{eq:spec-flow-rewrite}
    \Delta I_{\rm tot} = \frac{\SF}{4\pi}.
\end{align}
A picture of $D(\nu)$'s eigenvalues is in Figure~\ref{fig:spec-flow}.

\begin{figure}
    \centering
    \begin{tikzpicture}[xscale=3.4, yscale=0.62, >=stealth, line cap=round, nlab/.style={font=\footnotesize}]
        \draw[gray!60, thick] (-0.5,-0.5) -- (1.5,-2.5) node[right=2pt, nlab, gray!55!black] {$n=-1$};
        \draw[gray!60, thick] (-0.5, 2.5) -- (1.5, 0.5) node[right=2pt, nlab, gray!55!black] {$n=2$};
        \draw[blue!70!black, very thick, ->] (-0.5, 0.5) -- (1.5,-1.5) node[right=2pt, nlab, blue!70!black] {$n=0$};
        \draw[red!70!black,  very thick, ->] (-0.5, 1.5) -- (1.5,-0.5) node[right=2pt, nlab, red!70!black] {$n=1$};
        \node[gray!60!black] at (0.5, 2.45) {$\vdots$};
        \node[gray!60!black] at (0.5,-2.30) {$\vdots$};
        \draw[densely dashed, gray!55] (1,-2.6) -- (1,2.6);
        \draw[->, thick] (-0.85,0) -- (1.98,0) node[right] {$\nu$};
        \draw[->, thick] (0,-2.85) -- (0,2.85) node[above] {$\lambda$};
        \node[circle, fill=blue!70!black, inner sep=0pt, minimum size=3.6pt] at (0,0) {};
        \node[circle, fill=red!70!black,  inner sep=0pt, minimum size=3.6pt] at (1,0) {};
        \node[below left=-1pt] at (0,0) {$0$};
        \node[below right=1pt] at (1,0) {$1$};
    \end{tikzpicture}
    \caption{Spectral flow of the Dirac operator $D(\nu)$ with eigenvalues
    $\lambda_n = n - \nu$. As $\nu$ varies from $0$ to $1$, two eigenvalues
    touch zero: the $n = 0$ eigenvalue goes from zero to negative, and the
    $n = 1$ eigenvalue goes from positive to zero. Taking $\nu$ mod $1$,
    this can be seen as one eigenvalue crossing from positive to negative,
    thus producing a change in the $\eta$-invariant, $\Delta \eta_D = -2$.}
    \label{fig:spec-flow}
\end{figure}
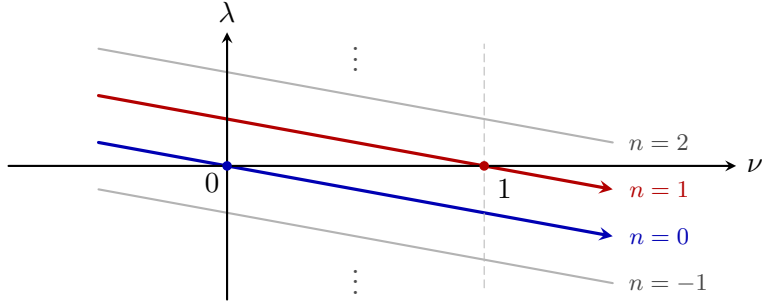

The current value itself is situation-dependent---we will see 
this in the finite-temperature section below---but the jump 
is robust. Physically, we can interpret the current jump as a right-moving 
angular momentum mode crossing the Fermi energy and becoming a left-moving angular
momentum mode.

Since $\partial_\nu \Omega = 4\pi I_{\rm tot}$ (recall that $\Omega$
is defined with a minus sign to its usual value, $\Omega = - \Omega_{\rm gc}$),
we can integrate (\ref{eq:reg-current}) to get
\begin{align}
    \label{eq:reg-omega}
    & \Omega(\nu) = \frac{1}{2} (\nu^2 - |\nu|),
    & |\nu| < 1.
\end{align}
$\Omega(\nu)$ is even, $\Omega(\nu) = \Omega(-\nu)$, but since
$\Omega(\nu)$ has a cusp at $\nu = 0$, we cannot conclude that
$\nu = 0$ is a critical point. When $\nu$ is small, (\ref{eq:reg-omega})
reduces to
\begin{align*}
    \Omega(\nu) \approx - \frac{|\nu|}{2}, 
\end{align*}
and thus the current is
\begin{align*}
    I_{\rm tot} \approx - \frac{\sgn(\nu)}{8\pi}.
\end{align*}
So, we can see that there is a persistent current
no matter how small the flux is! This is interesting, and
may be interpreted as the system being perfectly diamagnetic
at small values of $\nu$. This most likely only holds in the continuum
model, though, and thus should be considered in that light.

We again may gain insight into this phenomenon from the perspective 
of $\eta_D$: the basic idea is that $\eta_D$ is indifferent to
the size of $\nu$---it only cares about $\nu$ crossing an integer value.
Since $\Omega(\nu)$ is even, we must have opposite derivatives---and
thus currents---to the left- and right-hand sides of $\nu = 0$. The
requirement that $\Delta \eta_D = -2$ as $\Phi$ crosses zero from below
forces persistent currents at near-zero fluxes to provide the necessary jump
in the $\eta$-invariant. Physically, the current jump
is quantized because of the discrete nature of the number of left- vs. right-moving
particles in the system; turning the flux on is equivalent to turning a
right-mover into a left-mover---thus producing a discrete jump in the current.

\subsection{Finite-Temperature Disk}
\label{app:uv-finite-T}

To go from zero to finite temperature, we simply
must replace the sum $\sum_{E < 1}$ with $\sum_{E} n_D(E)$,
with $n_D(E)$ the Fermi-Dirac distribution. We have
\begin{align}
    I_{\rm tot}(\beta, \mu) = \int_0^\infty r \, \dd r \, j^\theta = \frac{1}{4b} \sum_{\ell \in \ints - \nu} \int_0^\infty \dd E \, n_D(E) k \frac{\dd N_\ell}{\dd E} \sgn(\ell).
\end{align}
Performing the integral gives
\begin{align}
    \label{eq:finite-temp-tot-current}
    I_{\rm tot}(\beta, \mu) = \frac{\eta_D}{8\pi} \left( \int_0^\infty \dd E \, n_D(E) \right) = \frac{\eta_D}{8\pi \beta} \log\left( 1 + e^{\beta \mu} \right).
\end{align}
At $\beta \to \infty$, $\mu = 1$, this reproduces the zero-temperature
current (\ref{eq:tot-current}). Every result of Section~\ref{app:uv-zero-T} 
holds in the finite-temperature case, but each is multiplied by the factor
\begin{align}
    \label{eq:finite-t-factor}
    c(\beta, \mu) = \frac{\log\left( 1 + e^{\beta \mu} \right)}{\beta}.
\end{align}
In the low-temperature limit, we have that
\begin{align}
    c(\beta, \mu) \approx \mu,
\end{align}
and so our results are simply multiplied by $\mu$
for large $\beta$. For $\mu = 1$, $\beta \to \infty$,
we clearly recover the results of Section~\ref{app:uv-zero-T}.
For small $\beta$ and fixed $\mu$,
\begin{align}
    \label{eq:high-temp-c}
    c(\beta, \mu) \simeq \frac{\log(2)}{\beta} + \frac{\mu}{2} + O(\beta).
\end{align}
So, the current diverges as $\beta \to 0 \iff T \to \infty$.
This is to be expected, as at fixed chemical potential $\mu$
with $T \to \infty$, each energy state is half-filled, and thus
the integral in (\ref{eq:finite-temp-tot-current}) diverges. Physically,
one can think of this as adding infinitely many left- and right-moving
particles to the system, but in asymmetric amounts. Specifically,
there is an asymmetry of $\eta_D(\nu) = 2\nu - 1$ particles per shell,
and so for an infinite number of filled shells, you get an infinite current
in one direction or the other.

Integrating $4\pi$ times (\ref{eq:finite-temp-tot-current}) with respect to $\nu$
gives
\begin{align}
    \label{eq:finite-temp-omega}
    \Omega(\beta, \mu) = \frac{(\nu^2 - |\nu|)}{2 \beta} \log\left( 1 + e^{\beta \mu} \right) = \frac{c(\beta, \mu)}{2}(\nu^2 - |\nu|).
\end{align}
Differentiating (\ref{eq:finite-temp-omega}) for small values of
$\nu$ and dividing by $4\pi$ produces
\begin{align}
    \label{eq:finite-temp-small-nu-current}
    I_{\rm tot}(\beta, \mu) \approx - \frac{\sgn(\nu)}{8\pi} \, c(\beta, \mu).
\end{align}
As $\nu$ crosses zero, we see the current jump
\begin{align}
    \Delta I_{\rm tot}(\beta, \mu) = \left(\lim_{\nu \to 0^+} I_{\rm tot}(\beta, \mu) \right) 
    - \left(\lim_{\nu \to 0^-} I_{\rm tot}(\beta, \mu) \right) 
    = -\frac{c(\beta, \mu)}{4\pi}.
\end{align}

However, for a physical system, it's more natural to hold the particle
density fixed on the disk. This produces a chemical potential that runs with
$T$, thus fixing the divergence seen in (\ref{eq:high-temp-c}).
Precisely, we would like to fix the particle density on the disk,
\begin{align}
    \label{eq:particle-density}
    n(\beta, \mu) = N(\beta, \mu)/\pi b^2,
\end{align}
where $N(\beta, \mu)$ is the particle number depending on $\beta$
and $\mu$. Calculating $N(\beta, \mu)$ and fixing $n(\beta, \mu) = n_0 = 1/4\pi$ 
implies that $\mu$ should run like
\begin{align}
    \label{eq:running-mu}
    \mu(\beta) = \frac{\log(e^{\beta} - 1)}{\beta}.
\end{align}
As $\beta \to 0$,
\begin{align}
    \label{eq:mu-infty}
    \mu(\beta \to 0) \simeq \lim_{\beta \to 0} \left( - \frac{\log(1/\beta)}{\beta} + \frac{1}{2} + O(\beta) \right) = - \infty.
\end{align}
This compensates for the divergence in the current,
as $c(\beta, \mu) = 1$, and so at fixed density,
eq. (\ref{eq:finite-temp-tot-current}) becomes
\begin{align}
    I_{\rm tot}(\beta, \mu) = \frac{\eta_D}{8\pi} c(\beta, \mu) = \frac{\eta_D}{8\pi},
\end{align}
regardless of the value of $\beta$. One way to think about this
is that the current is entirely dictated by the number of particles
in the system, and increasing $T$ and decreasing $\mu$ have antagonistic
effects in this regard. Increasing $T$ fills more energy levels with 
particles---thus increasing the particle number---while decreasing $\mu$ 
implies that its conjugate variable, $N$, should also decrease to minimize the 
grand canonical potential. These effects work to cancel each other out
to retain a constant particle density, thus keeping the current value constant.

\subsection{Radial Disorder}
\label{app:uv-radial}

We now show that the persistent current is robust to adding
radial potentials that are constant near the origin to the Hamiltonian. 
Let the Hamiltonian of our system be
\begin{align}
    \label{eq:radial-H}
    H(r) = \frac{1}{2}\left( p_r^2 + \frac{p_\theta^2}{r^2} \right) + V(r),
\end{align}
with the flux imposed through the twisted boundary condition
$\psi(r, \theta + 2\pi) = e^{- 2 \pi \ii \nu} \psi(r, \theta)$,
and with
\begin{align}
    \label{eq:allowed-potential}
    V(r) = V_0 + O(r),
\end{align}
as $r \to 0$. For example, $V(r) = \exp(-r^2/a^2)$ or
$V(r) = (1/2) \omega^2 r^2$. As non-examples, $V(r) = \kappa / r^2$ and $V(r) = \sin(1/r)$ 
would not satisfy the condition (\ref{eq:allowed-potential}). The time-independent 
Schrödinger equation is
\begin{align}
    \label{eq:tise}
    \left[ \frac{1}{2} \left( - \frac{1}{r} \partial_r (r \partial_r) - \frac{1}{r^2} \partial_\theta^2 \right) + V(r) \right] \psi(r, \theta) = E \psi(r, \theta).
\end{align}
Separating variables $\psi(r, \theta) \sim e^{\ii \ell \theta} R(r)$,
with $\ell \in \ints - \nu$, gives the radial equation
\begin{align}
    \label{eq:radial}
    R''(r) + \frac{1}{r} R'(r) + \left( 2(E - V(r)) - \frac{\ell^2}{r^2} \right) R(r) = 0.
\end{align}
Fix a self-adjoint extension at the origin by getting
rid of divergent solutions, that is, those solutions
to the radial equation (\ref{eq:radial}) that diverge as $r^{-\alpha} \to \infty$
for $r \to 0$ with $\alpha > 0$. For $\alpha = 0$, fix the extension by
discarding the logarithmic solution, $R(r) \sim \log(r)$.
Define $\alpha = |\ell|$, so $\alpha^2 = \ell^2$.
Then, we claim that the allowed solutions to (\ref{eq:tise})
are
\begin{align}
    \label{eq:radial-soln}
    R(r) = A_{\ell, E} r^\alpha \left( 1 + \frac{V_0 - E}{2(\alpha + 1)} r^2 + O(r^3) \right),
\end{align}
for $\alpha > 0$. This follows from the ansatz $R(r) = r^\alpha (c_0 + c_1 r + c_2 r^2 + O(r^3))$
with $c_0 =  A_{\ell, E}$ being substituted into (\ref{eq:radial})
with $r$ in a neighborhood near the origin, and then matching powers of $r$. 
Since the persistent current only depends on the near-origin
behavior of $\psi(r, \theta) \sim e^{\ii \ell \theta} R(r)$, adding a
potential (\ref{eq:allowed-potential}) can only change the value
of the persistent current, but not the persistence of said current.
Thus, the current persists at the origin, and the presence of 
its associated current jump is robust to the radial perturbations 
(\ref{eq:allowed-potential}). 

However, this really isn't a useful statement---one can simply tune 
$A_{\ell, E}$ to be arbitrarily small to essentially get rid of the current. 
The useful statement we can extract from the above considerations is that as 
you turn the flux off, $\Phi \to 0^+$, the current localizes as a $\delta$-function
at the origin.

To show this, let's first recall a few notions from measure theory. We
say a family of measures $\dd\mu_\alpha(r)$ \emph{vaguely converges}
to the measure $\dd\mu(r)$ if for all compactly supported test functions $f$,
we have
\begin{align}
    \label{eq:vague-conv}
    \lim_{\alpha \to 0^+} \int_0^\infty f(r) \, \dd\mu_\alpha(r) = \int_0^\infty f(r) \, \dd\mu(r).
\end{align}
We indicate vague convergence by writing $\dd\mu_\alpha(r) \rightharpoonup \dd\mu(r)$. 
The limit $\alpha \to 0^+$ is specific to the case at hand, but we will always
mean this limit when we write $\dd\mu_\alpha(r) \rightharpoonup \dd\mu(r)$.
It is simple to show that the family of measures
\begin{align}
    \label{eq:measures-conv}
    \alpha r^{2\alpha - 1} \mathbf{1}_{(0, R)}(r) \, \dd r \rightharpoonup \frac{1}{2} \delta_0,
\end{align}
where $\delta_0$ is the delta distribution measure, defined by
\begin{align}
    \label{eq:delta}
    \int_0^\infty f(r) \, \delta_0(r) = f(0).
\end{align}
Define the current measure as follows. First, define
an indicator function $\chi(r)$ that is $1$ for $r$
such that (\ref{eq:radial-soln}) holds and $0$ otherwise.
Roughly, this means that $\chi(r) = 1$ near the origin and 
$0$ otherwise. Then, recall that
\begin{align}
    \label{eq:recall-j}
    j^\theta(r) = \frac{1}{r^2} {\rm Im}\left(\psi^* \, \partial_\theta \psi\right) = \ell \, r^{2|\ell| - 2}|A_{\ell, E}|^2 \left(1 + O(r^2)\right),
\end{align}
near the origin, which follows from (\ref{eq:radial-soln}).
Specifically, near the origin only the $\ell = -\nu$ angular momentum
mode contributes. So actually
\begin{align}
    j^{\theta}(r) = - \nu \, r^{2|\nu| - 2} |A_{-\nu, E}|^2 \left( 1 + O(r^2) \right).
\end{align}
Then, define the current measure as
\begin{align}
    \dd I_{\nu, E}(r) = \chi(r) j^\theta \cdot (r \, \dd r),
\end{align}
and the difference in current measures as the flux is turned off as
\begin{align}
    \Delta I_{E}(r) = \left( \lim_{\nu \to 0^+} \dd I_{\nu, E}(r) \right) - \left( \lim_{\nu \to 0^-} \dd I_{\nu, E}(r) \right).
\end{align}
We claim that
\begin{align}
    \label{eq:delta-claim}
    \Delta I_{E}(r) = - |A_{0, E}|^2 \delta_0.
\end{align}
This follows from (\ref{eq:measures-conv}) with $\alpha = |\nu|$.
For $\nu \to 0^+$, we have
\begin{align}
    \dd I_{\nu, E}(r) \rightharpoonup - \frac{1}{2} |A_{0, E}|^2 \delta_0.
\end{align}
Similarly, as the flux is turned off, $\dd I_{\nu, E}(r)$ vaguely converges
to $(1/2) |A_{0, E}|^2 \delta_0$, thus giving (\ref{eq:delta-claim}). So, the current 
localizes as a delta function near the origin as the flux is turned off. Note that 
the total current can still be tuned to be arbitrarily small, as it's multiplied 
by a factor of $|A_{0, E}|^2$---the $\delta$-function convergence is the only
novel behavior that is truly robust to radial disorder.

\subsection{Persistent Current on the Lattice}
\label{app:uv-lattice}

There is the question of whether or not the persistent
current and its current jump appear on the lattice. 
As stated previously, since we're changing the UV
physics, we should expect the current to probably vanish.

To test this, we first need to define what the analog
of the persistent current is on the lattice. The analog of
\begin{align}
    \label{eq:continuum-current-int}
    I_{\rm tot} = \int r \, \dd r \, j^\theta,
\end{align}
on the lattice is defined by choosing a ray
from the origin (a ``slice''
of the lattice), and then considering the sum
of currents perpendicular to this slice,
\begin{align}
    I_{\rm tot} = \sum_{ij} I_{ij}.
\end{align}
The sum over $ij$ here sums over the nearest lattice
sites perpendicular to the slice we've chosen. For example,
summing over edges perpendicular to the ray $y = 0$, $x > 0$ 
gives the lattice analog of the integral (\ref{eq:continuum-current-int})---this
construction is pictured in Figure~\ref{fig:lattice-slice}.

\begin{figure}[H]
    \centering
    \begin{tikzpicture}[x=0.68cm, y=0.68cm, >=stealth, line cap=round, line join=round,
        bond/.style={gray!45, line width=0.5pt},
        site/.style={circle, fill=gray!30!black, inner sep=0pt, minimum size=2.6pt},
        cut/.style={red!70!black, very thick, ->, shorten >=1.5pt, shorten <=1.5pt},
        nlab/.style={font=\footnotesize}]
        \fill[red!18] (-0.5,-0.5) rectangle (0.5,0.5);
        \foreach \x in {-3.5,-2.5,...,2.5}
            \foreach \y in {-3.5,-2.5,...,3.5}{
                \pgfmathsetmacro{\ok}{(\x*\x+\y*\y<=16 && (\x+1)*(\x+1)+\y*\y<=16) ? 1 : 0}
                \ifdim\ok pt>0.5pt
                    \draw[bond] (\x,\y) -- (\x+1,\y);
                \fi
            }
        \foreach \x in {-3.5,-2.5,...,3.5}
            \foreach \y in {-3.5,-2.5,...,2.5}{
                \pgfmathsetmacro{\ok}{(\x*\x+\y*\y<=16 && \x*\x+(\y+1)*(\y+1)<=16) ? 1 : 0}
                \ifdim\ok pt>0.5pt
                    \draw[bond] (\x,\y) -- (\x,\y+1);
                \fi
            }
        \draw[gray!55, densely dashed] (0.30,0) -- (4.85,0)
            node[right, nlab, gray!45!black] {slice};
        \foreach \x in {0.5,1.5,2.5,3.5}
            \draw[cut] (\x,-0.5) -- (\x,0.5);
        \node[nlab, red!70!black] at (3.0,0.30) {$I_{ij}$};
        \foreach \x in {-3.5,-2.5,...,3.5}
            \foreach \y in {-3.5,-2.5,...,3.5}{
                \pgfmathsetmacro{\ok}{(\x*\x+\y*\y<=16) ? 1 : 0}
                \ifdim\ok pt>0.5pt
                    \node[site] at (\x,\y) {};
                \fi
            }
        \node[nlab, red!55!black] at (0,0) {$\Phi$};
        \draw[gray!40, densely dashed] (0,0) circle (4);
        \node[nlab, gray!45!black, fill=white, inner sep=1.5pt] at (135:4) {$R$};
    \end{tikzpicture}
    \caption{The physical setup on the lattice annulus,
    drawn for a disk of radius $R = 4$ in lattice units. The 
    AB solenoid threads the central plaquette with flux $\Phi$.
    The slice is the positive $x$-axis---the bonds crossing it 
    carry the currents contributing to $I_{\rm tot}$.}
    \label{fig:lattice-slice}
\end{figure}
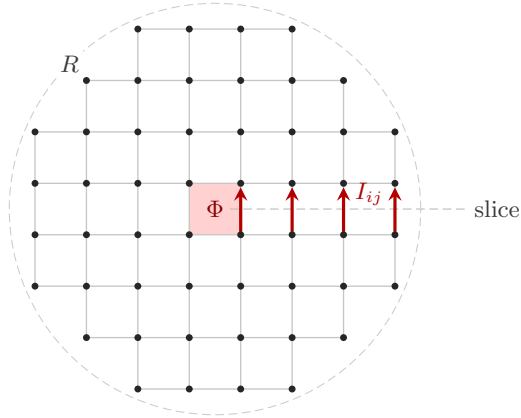

Using a slice along the positive $x$-axis as in Figure~\ref{fig:lattice-slice},
we compute the total current $I_{\rm tot} = \sum_{ij} I_{ij}$, using the same lattice
regularizations and parameter sweeps above. Since the current
is a short-distance phenomenon, we consider the sum up to a finite radius
$R_c$. We take $R_c = 6.5$ in lattice units.

\begin{figure}[H]
    \centering
    \includegraphics[width=0.85\textwidth]{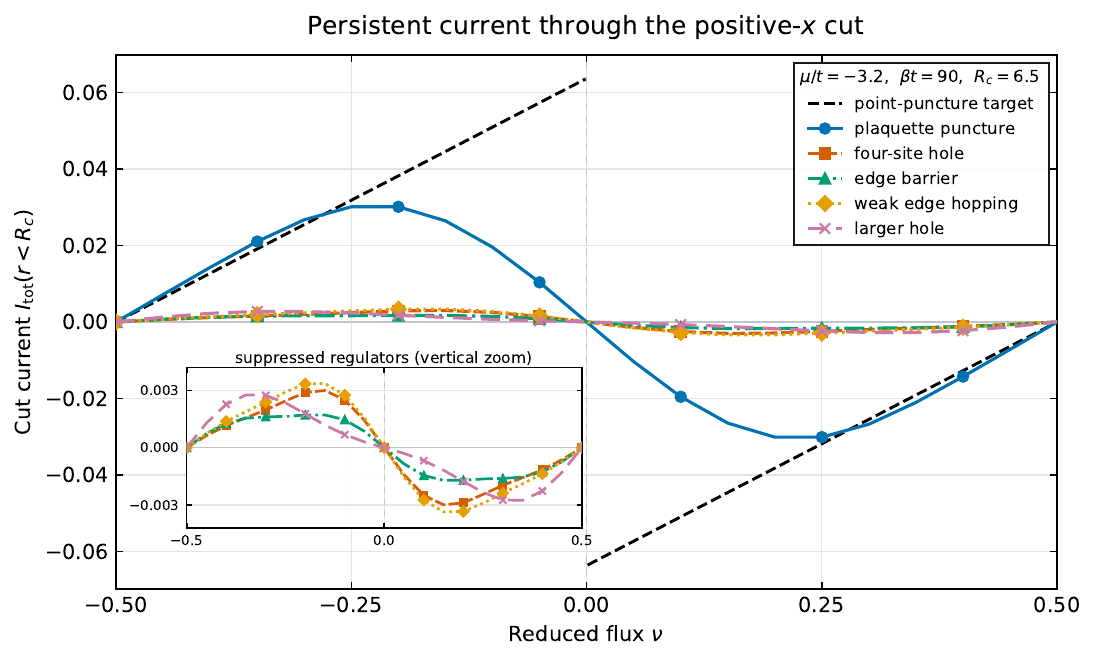}
    \caption{Total current $I_{\rm tot}(r < R_c)$ through the positive-$x$
    slice as a function of the reduced flux $\nu$, plotted for the five
    puncture regularizations at $\mu / t = - 3.2$ and $\beta t = 90$,
    with $R_c = 6.5$. The current is odd in $\nu$, as required from
    $\Omega_\beta(\nu, \mu) = \Omega_\beta(-\nu, \mu)$. As the inner
    annular radius becomes wider, the current value decreases.
    We can see that there is no persistent current at the origin
    and thus no associated current jump.}
    \label{fig:lattice-persistent-current}
\end{figure}

From inspection of Figure~\ref{fig:lattice-persistent-current}, it's clear that the UV regularization
destroys the persistent current and its associated current
jump. This fact holds for $195/200$ of the tested simulations.
The five simulations at $\mu = 0$ and $T = 0$ show very small persistent
currents, $|I(\nu \to 0^+)| \sim 8 \cdot 10^{-3}$, which is an order of magnitude
below the continuum value. Thus the continuum persistent current and its
current jump usually do not survive on the lattice, as was expected due to
their distinction as UV effects that depend on the details of the short-range
physics.

To check that this effect isn't only due to the outer boundary radius,
we also tested outer radii
\begin{align}
    b = 12, 16, 20, 24, 28, 32,
\end{align}
for $\mu/t = - 3.2$ for small values of $\nu$. The persistent current
still vanished in these cases. There is still the question of whether
or not the persistent current is present in the limit of $b \to \infty$.

\bibliographystyle{JHEP}
\bibliography{paperNotes}

\end{document}